%% file: main.tex
\documentclass[twocolumn]{autart}

\usepackage{amsmath, amssymb, mathtools}

\newcommand{\FFnorm}[1]{|\!|\!| #1 |\!|\!|}

\usepackage{amsmath} 
\usepackage{graphicx}
\usepackage{subcaption}
\usepackage{etoolbox}
\usepackage{etoolbox}
\usepackage{amsmath}
\newtheorem{assumption}{Assumption}  
\newtheorem{remark}{Remark}  
\usepackage{xcolor}
\usepackage{scalerel}
\usepackage{picins}
\usepackage{amssymb}      

\newenvironment{proof}{\noindent\textbf{Proof.}\ }{\hfill$\Box$\par}

\usepackage[htt]{hyphenat}
\usepackage{siunitx}
\usepackage{pifont}
\usepackage{graphicx}
\usepackage{amsmath,textcomp}
\usepackage{tikz}

\allowdisplaybreaks

\usepackage{url}
\usepackage{algpseudocode}
\usepackage{paralist}
\usepackage{color}
\usepackage{stmaryrd}
\usepackage{epsfig} 
\usepackage{bm} 
\usepackage{listings,lstautogobble}
\usepackage{mathtools}
\usepackage{mathrsfs}
\usepackage{amsfonts}
\usepackage{xspace}
\usepackage{verbatim}
\usepackage{epstopdf}
\usepackage{listings}
\usepackage{parcolumns}
\usepackage{color}
\usepackage{xspace}
\usepackage{verbatim}
\usepackage[utf8]{inputenc}
\usepackage{calrsfs}
\usepackage{rotating}
\usepackage{booktabs}
\usepackage{multirow}

\usepackage[compact]{titlesec}
\titlespacing{\section}{0pt}{1.4ex}{1.4ex}
\titlespacing{\subsection}{0pt}{0.8ex}{0.8ex}
\titlespacing{\subsubsection}{0pt}{0.4ex}{0.4ex}

\usepackage[T1]{fontenc}
\usepackage[utf8]{inputenc}

\newtheorem{definition}{Definition}
 \newtheorem{theorem}{Theorem}
  \newtheorem{lemma}{Lemma} 
  \newtheorem{property}{Property}
   \newtheorem{corollary}{Corollary}
\usepackage{algorithm}
\usepackage{algpseudocode}
  \renewcommand{\footnotesize}{\fontsize{8pt}{11pt}\selectfont}
\newtheorem{proposition}{Proposition}

\DeclareMathAlphabet{\mathcal}{OMS}{cmsy}{m}{n}

\usepackage{hyperref}
\hypersetup{
    colorlinks=true,
    linkcolor=black,
    urlcolor=gray,
}
\usepackage{balance}

\newcommand{\R}{{\mathbb{R}}}

\newcommand{\N}{{\mathbb{N}}}

\newcommand{\x}{\textbf{x}}

\newcommand{\zono}[1]{\langle #1 \rangle}

\makeatletter
\newcommand{\myitem}[1]{%
\item[#1]\protected@edef\@currentlabel{#1}}
\makeatother

\makeatletter
\newcommand{\vast}{\bBigg@{4}}
\newcommand{\Vast}{\bBigg@{5}}
\makeatother

\makeatletter
\DeclareRobustCommand{\nand}{\mathbin{\mathpalette\n@and@or\land}}
\DeclareRobustCommand{\nor}{\mathbin{\mathpalette\n@and@or\lor}}

\DeclareRobustCommand{\enand}{\overline{\mathbin{\mathpalette\n@and@or\land}}}
\DeclareRobustCommand{\enor}{\overline{\mathbin{\mathpalette\n@and@or\lor}}}

\newcommand{\n@and@or}[2]{%
  \vphantom{#2}%
  \ooalign{$\m@th#1#2$\cr\hidewidth$\m@th#1\sim$\hidewidth\cr}%
}
\makeatother

\begin{document}

\begin{frontmatter}

\title{Online Data-Driven Reachability Analysis using Zonotopic Recursive Least Squares}

\thanks[footnoteinfo]{Corresponding author Alireza Naderi Akhormeh.}

\author[tum]{Alireza Naderi}\ead{alireza.naderi@tum.de},    
\author[tum,guc]{Amr Hegazy}\ead{amr.hazem@student.guc.edu.eg}, and
\author[tum]{Amr Alanwar}\ead{alanwar@tum.de}

\address[tum]{Technical University of Munich, Germany\vspace{-2mm}} 

\address[guc]{German University in Cairo, Egypt\vspace{-2mm}}



\input{Sections/1-abs.tex}

\end{frontmatter}

\input{Sections/2-intro.tex}

\input{Sections/3-prelim}

\input{Sections/4-param}
\input{Sections/5-RLSReachability}

\input{Sections/6-evaluation}

\input{Sections/7-con}
\input{Sections/8-Acknowledgement}
\input{Sections/Appendix_A}
\input{Sections/Appendix_B}

\bibliographystyle{plain}
\bibliography{ref}

\end{document}

%% file: Sections/1-abs.tex
\begin{abstract}
Reachability analysis is a key formal verification technique for ensuring the safety of modern cyber–physical systems subject to uncertainties in measurements, system models (parameters), and inputs. Classical model-based approaches rely on accurate prior knowledge of system dynamics, which may not always be available or reliable. To address this, we present a data-driven reachability analysis framework that computes over-approximations of reachable sets directly from online state measurements. The method estimates time-varying unknown models using an Exponentially Forgetting Zonotopic Recursive Least Squares (EF-ZRLS) method, which processes data corrupted by bounded noise. Specifically, a time-varying set of models that contains the true model of the system is estimated recursively, and then used to compute the forward reachable sets under process noise and uncertain inputs. Our approach applies to both discrete-time Linear Time-Varying (LTV) and nonlinear Lipschitz systems. Compared to existing techniques, it produces less conservative reachable-set over-approximations, remains robust under slowly varying dynamics, and operates solely on real-time data without requiring any pre-recorded offline experiments. Numerical simulations and real-world experiments validate the effectiveness and practical applicability of the proposed algorithms.
\end{abstract}

%% file: Sections/2-intro.tex
\section{Introduction}
Reachability analysis in cyber-physical systems identifies all bounded reachable sets to formally avoid unsafe sets. Reachability computation is categorized as model-based when the analytical dynamics of the dynamic system are known, and data-driven when only state or output measurements are available. In model-based analysis, different approaches such as support function \cite{le2010reachability}, level set \cite{mitchell2000level}, and zonotope \cite{althoff2008reachability} are commonly used. Certain methods, such as level set techniques, can provide an exact representation of reachable sets for different system classes. However, these methods rely on predefined models, which are often derived from first principles. In many real-world dynamic systems, constructing an accurate model is non-trivial, which makes model-based reachability approaches challenging when an exact system model is unavailable.
\begin{figure}[thp]
    \centering
    \includegraphics[scale=0.55]{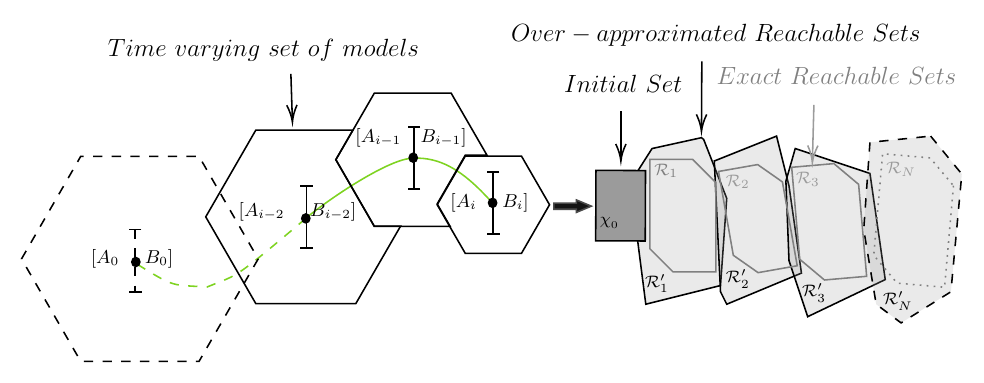}
    \caption{The reachable sets consistent with noisy input–state data for a time-varying dynamical system.}
    \label{fig:efzrls}
    \vspace{-1mm}
\end{figure}

Data-driven reachability analysis, unlike model-based approaches, enables the computation of over-approximate reachable sets directly from observed input-output data, even when an explicit system model is unavailable. This paradigm constructs a data-driven representation from recorded or real-time trajectories, which is then used to approximate the system’s reachable set. Recent advances leverage deep learning to solve the Hamilton-Jacobi (HJ) Partial Differential Equation (PDE) for high-dimensional systems. For instance, \cite{bansal2021deepreach} employs a self-supervised PDE loss to learn the value function offline, assuming known system dynamics (the Hamiltonian is analytically computed, not learned). In contrast, \cite{chilakamarri2024reachability} extends this framework to partially unknown dynamics by incorporating data-driven Hamiltonian approximation. Both methods operate offline but enable real-time reachability queries once trained. While these methods can compute more accurate reachable sets than other approaches, they require substantial training data and struggle with scalability in high-dimensional systems due to computational complexity.\par 

Using the zonotope, which can reduce computational cost in robust estimation, could be beneficial. Zonotopes have been applied in reachability analysis of logical systems \cite{ALANWAR2025111896}, hybrid logical and dynamical systems \cite{BIRD2023111107,COMBASTEL2022110457}, as well as state estimation and fault diagnosis \cite{REGO2025112380}. Consequently, zonotope-based approaches can compute more conservative reachable sets with less data and offer better scalability in higher-dimensional systems. In \cite{alanwar2023data}, reachable sets of discrete-time linear and nonlinear systems are computed using a zonotope-based approach by estimating a set of models from noisy input-state data. By collecting multiple trajectories, the method identifies a set of consistent models that account for known bounded noise corrupting the measurements. Although the method effectively handles uncertainty and model estimation from collected data using the least squares method and has been successfully applied to provide safety in control of nonlinear systems \cite{farjadnia2023robust}, it still requires pre-collected trajectories to construct the set of data-driven models. This makes it challenging when system models change over time steps or when there is difficulty in collecting offline data. To overcome model estimation challenges, recursive model estimation methods can be used to estimate and update models over time steps with new data. A set-based zonotopic recursive least squares method, introduced in \cite{samada2023zonotopic}, was designed for model estimation in a single-state (single-output regression model) Linear Time-Invariant (LTI) system.\par

The main idea of the proposed data-driven reachability analysis framework is illustrated in Figure \ref{fig:efzrls}. Our proposed method computes data-driven reachable sets using a recursive matrix zonotope formulation. To ensure that the reachable set encloses all possible system trajectories derived from finite noisy data, we construct a time-varying matrix zonotope that captures all models consistent with the noisy measurements, rather than relying on a single potentially inaccurate model. This is achieved through a set-based EF-ZRLS method, which estimates the models of the multi-output, time-varying regression models whose outputs are corrupted by bounded noise. The true system models are guaranteed to lie within these sets. Using the latest identified matrix zonotope that contains the true model, we propagate the initial state set forward to obtain the corresponding over-approximated reachable sets based on the ideas presented in \cite{alanwar2023data} for a discrete-time LTV system. Finally, we extend the framework to address a broader class of discrete-time nonlinear Lipschitz systems.\par

The main contributions of this paper are summarized as follows:

\begin{enumerate}
    
    \item A set-based EF-ZRLS method is proposed for estimating the unknown model of a time-varying multi-output regression model with measurements corrupted by bounded noise (Proposition \ref{pro:rze}). The optimality of the EF-ZRLS is established, and the corresponding optimal model estimation correction gain is derived (Theorem \ref{th:optimalgain}). The proposed EF-ZRLS formulation is applicable to both vector and matrix zonotopes.
    
    \item Algorithm \ref{alg:LTIreach1} is introduced to compute over-approximated reachable sets (Theorem \ref{th:reach_lin}) for discrete-time LTV systems \eqref{eq:sys} subject to bounded process noise. 

    \item Algorithm \ref{alg:LipReachability} is introduced to compute over-approximated reachable sets (Theorem \ref{th:reachdisnonlin}) for a class of discrete-time nonlinear Lipschitz systems \eqref{eq:sysnonlingen} subject to bounded process noise.
    
    \item A real-world experiment, two numerical examples, and a comparative analysis are presented, contrasting the proposed data-driven reachability analysis with the alternative model-based approach.

\end{enumerate}

To recreate our results, readers can use our publicly available 
code\footnote{\href{https://github.com/TUM-CPS-HN/ZRLS}{https://github.com/TUM-CPS-HN/ZRLS}} 
and a video of experimental results\footnote{\href{https://www.youtube.com/watch?v=6va1PpCKL9A}{https://www.youtube.com/watch?v=6va1PpCKL9A}}.

 The remainder of the paper is organized as follows. Section \ref{sec:prelim} presents the problem formulation and relevant preliminaries. In the subsequent section, we propose a robust recursive set-based estimation method to estimate the unknown models of a time-varying regression model corrupted by noisy data. In Section \ref{sec:RLSReachability}, we introduce our proposed data-driven reachability analysis method for LTV and nonlinear Lipschitz continuous dynamical systems. Section \ref{sec:evauation}, evaluates the effectiveness of the proposed method through numerical examples and an experimental case study. Finally, Section \ref{sec:con} concludes the paper and discusses potential directions for future work.

%% file: Sections/3-prelim.tex
\section{Problem Formulation and Preliminaries }\label{sec:prelim}
In this section, we present the problem statement, the preliminary definitions for reachability analysis, and introduce details about the notation used throughout this work.
\subsection{Problem Formulation}
We consider a discrete-time system
\begin{align}
\begin{split}
    x_{k+1} &= f(x_k,u_k)+ w_k.\\
\end{split}
    \label{eq:sysnonlingen}
\end{align}
where $f:\mathbb{R}^{n_x}\times\mathbb{R}^{n_u} \rightarrow \mathbb{R}^{n_x}$ a differentiable unknown function, $w_k \in \mathcal{Z}_w \subset \mathbb{R}^{n_x}$ denotes the noise bounded by a noise zonotope $\mathcal{Z}_w$, ${u_k \in \mathcal{U}_k}$ the input bounded by an input zonotope $\mathcal{U}_k \in \mathbb{R}^{n_u}$, $x_k\in\mathbb{R}^{n_x}$ the measured state, and ${x_0 \in \mathcal{X}_0 \subset \mathbb{R}^{n_x}}$ the initial state $x_0$ of the system bounded by the initial set $\mathcal{X}_0$.\par 
Reachability analysis computes the set of states $x_k$ that can be reached given a set of uncertain initial states $\mathcal{X}_0$ and a set of uncertain inputs $\mathcal{U}_k$. More formally, it can be defined as follows:
\begin{definition} (\textbf{Exact Reachable Set})
The exact reachable set $\mathcal{R}_{N}$ after $N$ time steps subject to inputs ${u_k \in \mathcal{U}_k}$, $\forall k \in\{ 0, \dots, N-1\}$, and noise $w_k \in \mathcal{Z}_w$, is the set of all states trajectories starting from initial set $\mathcal{X}_0$ after $N$ steps: 
\begin{align} \label{eq:R}
        \mathcal{R}_{N} = \big\{& x_N \in \mathbb{R}^{n_x} \, \big| x_{k+1} = f(x_k,u_k) + w_k, x_0 \in \mathcal{X}_0, \nonumber\\
       & \,  u_k \in \mathcal{U}_k, w_k \in \mathcal{Z}_w:  \forall k \in \{0,...,N{-}1\}\big\}.
\end{align}
\end{definition}
We consider a sliding window single input-state trajectory with lengths $N_D$, denoted by $\{u_i\}$, and $\{x_i\}$, for all $i\ge0$. 
Let us denote the sliding window using the shifted signals as follows
\begin{align*}
X_+ &= \begin{bmatrix}x_{i-N_D+1},x_{i-N_D+2}, \dots, x_{i} \end{bmatrix}, \\
X_- &= \begin{bmatrix}x_{i-N_D},x_{i-N_D+1}, \dots, x_{i-1} \end{bmatrix},\\
U_- &= \begin{bmatrix} u_{i-N_D},u_{i-N_D+1}, \dots, u_{i-1} \end{bmatrix}.
\end{align*}
We denote the set of all available data by $D_{-}=(X_-,U_-)$ and $D=(X_+,X_-,U_-)$. 

We aim to compute an over-approximation of the exact reachable sets when the model of the system in~\eqref{eq:sysnonlingen} is unknown, but input and noisy states are available at each time step.\par

\subsection{Notation}
The set of real and natural numbers are denoted as $\mathbb{R}$ and $\mathbb{N}$, respectively, and $\mathbb{N}_0 = \mathbb{N} \cup \{0\}$. 
The $2$-norm of a vector $\bar{x} \in \mathbb{R}^{\bar{n}}$ is defined as $\|\bar{x}\|_2 := \sqrt{\sum_{i=1}^{\bar{n}} |\bar{x}_i|^2}$. 
For a matrix $A \in \mathbb{R}^{n_A \times m_A}$, the Frobenius norm is defined as $\|A\|_F := \sqrt{\sum_{i=1}^{n_A} \sum_{j=1}^{m_A} |a_{ij}|^2} = \|\mathrm{vec}(A)\|_2$, and the matrix infinity norm is defined as $\|A\|_{max} := \max_{i,j}|a_{ij}|$. 
The transpose and Moore–Penrose pseudoinverse of a matrix $X$ are denoted as $X^\top$ and $X^\dagger$, respectively. 
We denote the Kronecker product by $\otimes$, the element at row $i$ and column $j$ of matrix $A$ by $(A)_{i,j}$, and column $j$ of $A$ by $(A)_{.,j}$. 
The vectorization of a matrix $A$ is defined by $\mathrm{vec}(A) \in \mathbb{R}^{n_A m_A}$, obtained by stacking the columns of $A$ into a single column vector. Conversely, the inverse operation $\mathrm{unvec}(\bar{x})$ reshapes a vector $\bar{x} \in \mathbb{R}^{n_A m_A}$ back into a matrix in $\mathbb{R}^{n_A \times m_A}$. For a list or vector of elements, we denote the element $i$ of vector or list $a$ by $a^{(i)}$. The element-wise (Hadamard) product of two matrices is denoted by $\odot$. 
We use $\mathcal{Z}_1 \ominus \mathcal{Z}_2$ to denote $\mathcal{Z}_1 \oplus (-1) \mathcal{Z}_2$, which is not the Minkowski difference. 
We also define for $N$ time steps
\begin{align*}
    \mathcal{F} = \cup_{k=0}^{N} (\mathcal{R}_k \times \mathcal{U}_k).
    \label{eq:F}
\end{align*}

\subsection{Set Representations}
We start by defining some set representations that are used in the reachability analysis. 

\begin{definition}(\textbf{Zonotope} \cite{conf:zono1998}) \label{def:zonotopes} 
Given a center $c_{\mathcal{Z}} \in \mathbb{R}^{n_\mathcal{Z}}$ and $\gamma_{\mathcal{Z}} \in \mathbb{N}$ generator vectors in a generator matrix $G_{\mathcal{Z}}=\begin{bmatrix} g_{\mathcal{Z}}^{(1)}& \dots &g_{\mathcal{Z}}^{(\gamma_{\mathcal{Z}})}\end{bmatrix} \in \mathbb{R}^{n_\mathcal{Z} \times \gamma_{\mathcal{Z}}}$, a zonotope is defined as
\begin{equation}\label{eq:vector_zonotope}
	\mathcal{Z} = \Big\{ x \in \mathbb{R}^{n_\mathcal{Z}} \; \Big| \; x = c_{\mathcal{Z}} + \sum_{i=1}^{\gamma_{\mathcal{Z}}} \beta^{(i)} \, g^{(i)}_{\mathcal{Z}} \, ,
	\left|\beta^{(i)}\right| \leq 1 \Big\} \; .
\end{equation}
We use the shorthand notation $\mathcal{Z} = \zono{c_{\mathcal{Z}},G_{\mathcal{Z}}}$ for a zonotope. 
\end{definition}
Let $L \in \mathbb{R}^{m_L \times n_\mathcal{Z}}$ be a linear map. Then $L\mathcal{Z}= \zono{Lc_{\mathcal{Z}},LG_{\mathcal{Z}}}$ \cite[p.18]{conf:thesisalthoff}.

Given two zonotopes $\mathcal{Z}_1=\langle c_{\mathcal{Z}_1},G_{\mathcal{Z}_1} \rangle$ and $\mathcal{Z}_2=\langle c_{\mathcal{Z}_2},G_{\mathcal{Z}_2} \rangle$, the Minkowski sum $\mathcal{Z}_1 \oplus \mathcal{Z}_2 = \{z_1 + z_2| z_1\in \mathcal{Z}_1, z_2 \in \mathcal{Z}_2 \}$ can be computed exactly as follows \cite{conf:zono1998}: 
\begin{equation}
     \mathcal{Z}_1 \oplus \mathcal{Z}_2 = \Big\langle c_{\mathcal{Z}_1} + c_{\mathcal{Z}_2}, [G_{\mathcal{Z}_1}, G_{\mathcal{Z}_2} ]\Big\rangle.
     \label{eq:minkowski}
\end{equation}
We define and compute the Cartesian product of two zonotopes $\mathcal{Z}_1 $ and $\mathcal{Z}_2$ by 
\begin{align}\label{eq:cart}
\mathcal{Z}_1 \times \mathcal{Z}_2 &= \bigg\{ \begin{bmatrix}z_1 \\ z_2\end{bmatrix} \bigg| z_1 \in \mathcal{Z}_1, z_2 \in \mathcal{Z}_2 \bigg\} \nonumber\\
&= \Bigg\langle \begin{bmatrix} c_{\mathcal{Z}_1} \\ c_{\mathcal{Z}_2} \end{bmatrix}, \begin{bmatrix} G_{\mathcal{Z}_1} & 0 \\ 0 & G_{\mathcal{Z}_2}\end{bmatrix} \Bigg\rangle.
\end{align}

\begin{definition}\label{def:matzonotopes}(\textbf{Matrix Zonotope} \cite[p.52]{conf:thesisalthoff})  
Given a center matrix $C_{\mathcal{M}} \in \mathbb{R}^{n_\mathcal{M} \times m_\mathcal{M}}$ and $\gamma_{\mathcal{M}} \in \mathbb{N}$ generator matrices $\check{G}_{\mathcal{M}}=\begin{bmatrix}
     G_{\mathcal{M}}^{(1)}\dots G_{\mathcal{M}}^{(\gamma_{\mathcal{M}})} \end{bmatrix} \in \mathbb{R}^{n_\mathcal{M} \times (m_\mathcal{M} \gamma_{\mathcal{M}})}$, a matrix zonotope is defined as
\begin{equation}\label{eq:matrix_zonotope}
	\mathcal{M} {=} \Big\{ X \in \mathbb{R}^{n_\mathcal{M} \times m_\mathcal{M}} \; \Big| \; X {=} C_{\mathcal{M}} + \sum_{i=1}^{\gamma_{\mathcal{M}}} \beta^{(i)} \, G^{(i)}_{\mathcal{M}} \, ,
	\left|\beta^{(i)}\right| \leq 1\Big\} .
\end{equation}
We use the shorthand notation $\mathcal{M} = \zono{C_{\mathcal{M}},\check{G}_{\mathcal{M}}}$ for a matrix zonotope. 
\end{definition}

A matrix zonotope can be over-approximated by an interval matrix $I_\mathcal{M} =\operatorname{int(}{\mathcal{M}}) =  [\underline{\mathcal{M}},\overline{\mathcal{M}}]$, where $\underline{\mathcal{M}} \in \R^{n_\mathcal{M} \times m_\mathcal{M}}$ and $\overline{\mathcal{M}} \in \R^{n_\mathcal{M} \times m_\mathcal{M}}$ denote its lower bound and upper bound, respectively. The Frobenius norm of the interval matrix $I_\mathcal{M}$ is defined as
$\FFnorm{I_\mathcal{M}} = \|{|I_{\mathcal{M}_c}|+\Delta}\|_F$ with $I_{\mathcal{M}_c} = 0.5(\overline{\mathcal{M}}+\underline{\mathcal{M}})$ and $\Delta = 0.5(\overline{\mathcal{M}}-\underline{\mathcal{M}})$ \cite{farhadsefat2011norms}.\par
By construction, any matrix zonotope (Definition \ref{def:matzonotopes}) can be put in one-to-one correspondence with a vector zonotope (Definition \ref{def:zonotopes}), via
\[
\mathcal{Z} = \{\operatorname{vec}(M)\mid M\in\mathcal{M}\}, 
\quad 
\mathcal{M} = \{\operatorname{unvec}(Z)\mid Z\in\mathcal{Z}\},
\]
with centers and generators related by
\[
c_{\mathcal{Z}} = \operatorname{vec}(C_{\mathcal{M}}), 
\quad g_{\mathcal{Z}_i} = \operatorname{vec}(G_{\mathcal{M}_i}), 
\quad \forall i = 1, \dots, \gamma_{\mathcal{M}}.
\]

%% file: Sections/4-param.tex
\section{Robust Adaptive Recursive Parameter Estimation}\label{sec:param}
A robust adaptive parameter (model) estimation method is proposed for a linear regression model with noisy measurements and time-varying parameters. The approach utilizes a recursive, set-based method to iteratively estimate a time-varying set of parameters over consecutive time steps, which has been introduced in \cite{samada2023zonotopic} for a single-output time-invariant linear regression model. We extend it to the time-varying multi-output regression model. 
\subsection{Exponentially Forgetting Zonotopic Recursive Least Squares}
We consider a discrete-time system in regression form given by
\begin{align}\label{eq:regression}
\begin{split}
  y_k &= \varphi_k \theta_{\mathrm{tr},k} + v_k, \\
  \theta_{\mathrm{tr},k} &= \theta_{\mathrm{tr},k-1} + \delta \theta_{k-1}.
  \end{split}
\end{align}
where \( y_k \in \mathbb{R}^{p \times m} \) is the measured output, \( \varphi_k \in \mathbb{R}^{p \times n} \) is the regressor matrix at time step \(k\), and \( \theta_{\text{tr},k} \in \mathbb{R}^{n \times m} \) is the true time-varying (unknown) parameter matrix, and the term $\delta\theta_{k-1} \in \mathbb{R}^{n\times m}$ represents the change in the parameters. 
\begin{assumption}[\textbf{Bounded Noise}]\label{ass:noisebound}
    The additive measurement noise \( v_k \in \mathbb{R}^{p \times m} \) is bounded as \( \|v_k\|_{max} \le \sigma_v \), with \( \sigma_v \in \mathbb{R}_{\geq 0} \) a known scalar for all $k \geq 0$. To describe the uncertainty structure of the noise, we define the generator matrix \( G_v \in \mathbb{R}^{p \times pm} \) as \( G_v = \begin{bmatrix}
    Q_v^{(1)}, Q_v^{(2)}, \dots, Q_v^{(pm)} \end{bmatrix}  \), where each \( Q_v^{(l)} \in \mathbb{R}^{p \times m} \) contains exactly one nonzero element equal to \( \sigma_v \), and all other entries are zero. Specifically, for each \( l \in \{1, \dots, pm\} \), there exists a unique index pair \( (i, j) \) such that $(Q_v^{(l)})_{i,j} = {\sigma_v}, \quad \text{and} \quad (Q_v^{(l)})_{i_l,j_l} = 0 \quad \forall~ (i_l, j_l) \neq (i, j)$. Then, the noise satisfies \( v_k \in \zono{0, G_v} \) for all \(k\).
\end{assumption}

\begin{assumption}[\textbf{Bounded Parameter Variations}]\label{ass:parametebound}
    There exists ${\sigma}_\theta \geq 0$ such that, $\|{\delta\theta_{k-1}}\|_{max} \leq {\sigma}_\theta$. In addition to describe the uncertainty structure of parameter changing, we define the generator matrix \( G_{\theta} \in \mathbb{R}^{n \times nm} \) as \( G_{\theta} = \begin{bmatrix}
    G_{\theta}^{(1)}, G_{\theta}^{(2)}, \dots, G_{\theta}^{(nm)} \end{bmatrix} \), where each \( G_{\theta}^{(l)} \in \mathbb{R}^{n \times m} \) contains exactly one nonzero element equal to \( {{\sigma}_\theta} \), and all other entries are zero. Specifically, for each \( l \in \{1, \dots, nm\} \), there exists a unique index pair \( (i, j) \) such that $(G_{\theta}^{(l)})_{i,j} = {{\sigma}_\theta}, \quad \text{and} \quad (G_{\theta}^{(l)})_{i_l,j_l} = 0 \quad \forall~ (i_l, j_l) \neq (i, j)$. Then, the uncertainaty satisfies \( \delta_{{\theta}_{k-1}} \in \zono{0, G_{\theta}} \) for all \(k\).\par
\end{assumption}
To estimate unknown parameters, we need the regressor matrix sequence to be sufficiently rich. Consequently, we assume the sequence \(\{\varphi_k\}_{k=0}^\infty\) to be Persistently Exciting (PE) as defined in Definition~\ref{def:PE}.
\begin{definition}[\textbf{Persistently Exciting}]\label{def:PE}
The sequence $\{\varphi_k\},k=0,1,2,...\infty$ is said to be PE if for some constant $S\geq1$ and all $j$ there exist positive constants $\alpha$ and $\beta$ such that
\begin{align*}
    0<\alpha I_n \leq \sum_{i=j}^{j+S-1}\varphi^\top_i\varphi_i\leq\beta I_n<\infty.
\end{align*}
\end{definition}
Suppose the unknown true parameter matrix \(\theta_{\mathrm{tr},k}\) must be conservatively estimated. Under PE condition on \(\varphi_k\), the goal is to recursively compute a zonotope that over-approximates \(\theta_{\mathrm{tr},k}\), such that
\begin{equation}\label{eq:ZRLS_theta}
\theta_{\mathrm{tr},k+1} \in\zono{C_{k+1}, G_{\lambda_{k+1}}}.
\end{equation}
where $C_{k+1} \in \mathbb{R}^{n \times m}$ is the center matrix, and $G_{\lambda_{k+1}}^{(i)} \in \mathbb{R}^{n \times m}$ is the $i$-th generator matrix for all $i=1, \dots, r$. The number of generators is denoted by $r \in \mathbb{N}_{>0}$. We define recursive estimator as follows
\begin{align}\label{eq:thetaHatRecursion}
\theta_{k+1}= \theta_{k}+ K_k(y_{k}-\varphi_{k}\theta_{k} -v_k).
\end{align}
such that $\theta_{k+1}\in\R^{n\times m}$ provides the updated estimate when a new measurement is added. By multiplying the correction gain $ K_k\in\R^{n\times p}$  with the output estimation error $(y_{k}-\varphi_{k}\theta_{k} -v_k)$ , the previous estimation $\theta_{k}$ is updated at each sample time $k$. Inspired by \cite{alanwar2023distributed,le2013zonotope} of using set-propagation for designing diffusion observer and \cite{samada2023zonotopic} to parameter estimation of a regression model, we propose the Proposition \ref{pro:rze} to estimate a recursive set of unknown time-varying parameters in \eqref{eq:regression}.

\begin{proposition}[\textbf{Recursive Zonotopic Estimator}] \label{pro:rze}
Consider the deterministic system in~\eqref{eq:regression} with initial parameter set 
\(\theta_{0} \in \zono{C_{0}, G_0}\) where $\operatorname{rank}(G_0)\geq nm$ and choosing $K_k$ such that $\operatorname{rank}(I_n - K_k \varphi_k)=n$.
The EF-ZRLS is updated recursively by \eqref{eq:ZRLS_theta} for all $k=0,1,\dots\infty$ as follows
\begin{align}
C_{k+1} &= (I_n - K_k \varphi_k) C_k + K_k y_k, \label{eq:C_update}
\end{align}
\begin{subequations}\label{eq:R_update}
\begin{align}
G^{(i)}_{\lambda_{k+1}} &= \lambda^{-\frac12}(I_n - K_k \varphi_k)\, G^{(i)}_{\lambda_k}, \label{eq:R_update_a}\\
G_{\lambda_{k+1}} &= 
\begin{bmatrix}
G_{\lambda_{k+1}}^{(1)}, \dots, G_{\lambda_{k+1}}^{(N_G)}, G_{v_k}
\end{bmatrix},~\forall i=0,1,\dots,N_G. \label{eq:R_update_b}
\end{align}
\end{subequations}
where $N_G\in\N$ is the number of generators, \(\lambda \in (0,1]\) is the forgetting factor, and for all $p, m \in \mathbb{N}_{>0}$
\begin{align*}
    &G_{v_k} = \begin{bmatrix}
    -K_k Q_{v}^{(1)},\; -K_k Q_{v}^{(2)},\; \dots,\; -K_k Q_{v}^{(pm)} \end{bmatrix}.
\end{align*}
\end{proposition}
\begin{proof}
We can rewrite the \eqref{eq:thetaHatRecursion} as follows
\begin{align}\label{eq:thetaHatRecursion2}
\theta_{k+1}=(I_n - K_k \varphi_k)\theta_{k}+ K_ky_{k}+ K_kv_k.
\end{align}
Considering \eqref{eq:regression}, the estimated parameter contained the true value and its deviation as $\theta_k$. Given ${\theta}_k \in \zono{C_{k},G_{k}}$ and according to \eqref{eq:thetaHatRecursion2} 
\begin{align*}
\theta_{k+1} \in\, & 
\left( (I_n - K_k \varphi_k) \odot 
\zono{C_{k}, G_k } \right) 
\\ & 
\oplus \left( (I_n - K_k \varphi_k) \odot 
\zono{\textbf{0}, G_{\theta} } \right)  
\\ & 
\oplus \left( K_k \odot \zono{y_k, \textbf{0}} \right) 
\oplus \left( -K_k \odot \zono{\textbf{0}, G_v} \right).
\end{align*}
 By applying the linear map and Minkowski sum properties of matrix zonotopes, the above equation leads to the center update equation \eqref{eq:C_update} and the generator update equation as follows
 \begin{align*}
     G_{k+1} = &  \begin{bmatrix}
         (I_n - K_k \varphi_k) G_k, (I_n - K_k \varphi_k)G_{\theta},G_{v_K}\end{bmatrix}.
\end{align*}
Given the matrix zonotope Definition \ref{def:matzonotopes}, we can write   
\begin{align}\label{eq:M1zon1}
    \zono{C_{k+1},G_{k+1}}&=C_{k+1} + \sum_{j=1}^{N_G} \alpha_{j} (I_n - K_k \varphi_k)G^{(j)}_{k}\\\nonumber 
    &\quad + \sum_{i=1}^{nm} \beta_{i} (I_n - K_k \varphi_k)G^{(i)}_{\theta}+ \sum_{j=1}^{mp} \xi_{j} G^{(j)}_{v_k}, \\\nonumber 
    &\quad 
      \; \left|\alpha_{j}\right| \leq 1,\;
    \left|\beta_{i}\right| \leq 1,\;
    \left|\xi_{j}\right| \leq 1.
\end{align}
With regards to the definition of the proposition, we have $G^{(i)}_{\theta}\in \operatorname{span}(G_k)$  which means $G^{(i)}_{\theta} = \sum_{j=1}^{N_G} \eta_{ij}G^{(j)}_k$ for all $i=1,2,...,nm$ and $\eta_{ij}>0$. Then substituting it into \eqref{eq:M1zon1}, we have 
\begin{align}\label{eq:updateZonop1}
     \zono{C_{k+1},G_{k+1}}&=C_{k+1}+ \sum_{j=1}^{mp} \xi_{j} G^{(j)}_{v_k}\\\nonumber
     &\quad +\sum_{j=1}^{N_G}(I_n - K_k \varphi_k) (\alpha_{j} +\beta_{j} \sum_{j=1}^{nm}\eta_{ij})G^{(j)}_k,\\\nonumber
    &\quad 
      \;\left|\alpha_{j}\right| \leq 1,\;
    \left|\beta_{i}\right| \leq 1,\;
    \left|\xi_{j}\right| \leq 1.
\end{align}
If $\lambda^{\frac{-1}{2}} \gamma_{j}=\alpha_{j} + \beta_{j} \sum_{j=1}^{nm}\eta_{ij}$ for any $\lambda>0$. Then, using the triangle inequality, we have 
\begin{align*}
      \lambda^{\frac{-1}{2}} \left|\gamma_{j}\right|\leq\left|\alpha_{j}\right| + \left|\beta_{j}\right|
      \sum_{j=1}^{nm}\left|\eta_{ij}\right| . 
\end{align*}
Using the fact that $\left|\alpha_{j}\right| \leq 1$ and $\left|\beta_{j}\right| \leq 1$, we have 
\begin{align*}
       \left|\gamma_{j}\right|\leq (1 + 
      \sum_{j=1}^{nm}\left|\eta_{ij}\right|)\lambda^{\frac{1}{2}} . 
\end{align*}
Then by choosing 
\begin{align}\label{eq:lamda}
       0<\lambda\leq \frac{1}{(1 + 
      \sum_{j=1}^{nm}\left|\eta_{ij}\right|)^2}\leq1 . 
\end{align}
leads to $ \left|\gamma_{j}\right| \leq 1$. Therefore we can rewrite \eqref{eq:M1zon1} as follows
\begin{align}\label{eq:updateZonop1}
     \zono{C_{k+1},G_{\lambda_{k+1}}}&=C_{k+1} + \sum_{j=1}^{N_G}\lambda^{\frac{-1}{2}}(I_n - K_k \varphi_k) \gamma_{j} G^{(j)}_{k}\\\nonumber
     &\quad + \sum_{j=1}^{mp} \xi_{j} G^{(j)}_{v_k},
    \quad \left|\gamma_{j}\right| \leq 1.
\end{align}
Using matrix zonotope Definition \ref{def:matzonotopes} and zonotope containment condition in \cite[Theorem 3]{sadraddini2019linear} leads to $\zono{C_{k+1},G_{k+1}}\subseteq\zono{C_{k+1},G_{\lambda_{k+1}}}$. 
Therefore, the generator update \eqref{eq:updateZonop1} is written as follows 
\begin{align*}
     G^{(i)}_{{\lambda}_{k+1}} = &  \begin{matrix} [\lambda^{\frac{-1}{2}}(I_n - K_k \varphi_k) G^{(i)}_{{\lambda}_k}, G_{v_k}]\end{matrix}.
\end{align*}
\end{proof}
While in the Proposition \ref{pro:rze}, we proposed the center and generators update equations to find a set containing true parameters $\theta_{\mathrm{tr},k}$, the correction gain $K_k$ is still unknown in \eqref{eq:C_update}
 and \eqref{eq:R_update_a}. Theorem \ref{th:optimalgain} proposes an optimal gain that minimizes the estimated parameter set.
\begin{theorem}[\textbf{Optimal Estimator Gain}]\label{th:optimalgain}
Given the generator matrix \( G_{{\lambda}_{k+1}} \) as in \eqref{eq:R_update}, and under PE condition of $\varphi_k$ as specified in Definition~\ref{def:PE}, the optimal estimator-gain \( K^{\star}_{k} = \arg\min_{K_k} \|G_{{\lambda}_{k+1}}\|^2_{F,W} \) that minimizes the trace of the covariance matrix \( J_W = \text{Tr}(W P_{k+1}) \), is computed by
\begin{align}
K^{\star}_{k} &= {P}_k \varphi^\top_k \Lambda^{-1}, \label{eq:Correctiongain} \\
P_{k+1} &= \lambda^{-1}(I_n - K^{\star}_{k} \varphi_k) {P}_k \label{eq:Cov2}
\end{align}
where
\[
 \quad \Lambda = \varphi_k {P}_k \varphi^\top_k + \lambda Q,
\]
with \( Q = m\sigma^2_v\mathbf{1}
\in\R^{p\times p} \), and \( K^{\star}_{k}\) is independent of the weighting matrix \( W\in\R^{n\times n} \).
\end{theorem}
\begin{proof}
 Given the covariance Definition \ref{de:CovMatrixZono} , we substitute \eqref{eq:R_update} into the covariance matrix $P_{k+1}= G_{{\lambda}_{k+1}} G^{\top}_{{\lambda}_{k+1}} $ as follows
\begin{align}\label{eq:cov0}
P_{k+1} &= \lambda^{-1}(I_n - K_k \varphi_k) G_{{\lambda}_{k}} G^{\top}_{{\lambda}_{k}}
(I_n - K_k \varphi_k)^\top \\\nonumber
&\quad + K_k\sum_{j=1}^{pm} Q_{v}^{(j)} Q_{v}^{(j)^\top} K^\top_k .
\end{align}
Given $P_k=G_{{\lambda}_{k}} G^{\top}_{{\lambda}_{k}}$ and define \( Q =\sum_{j=1}^{pm} Q_{v}^{(j)} Q_{v}^{(j)^\top}\), \eqref{eq:cov0} is written as follows
\begin{equation}\label{eq:Cov1}
P_{k+1} = \lambda^{-1}(I_n - K_k \varphi_k) {P}_k (I_n - K_k \varphi_k)^\top  + K_kQK^\top_k.
\end{equation}
With replacing $P_{k+1}$ in $J_W$, it derives in
\begin{equation}
J_W = \text{Tr}(\lambda^{-1}W(I_n - K_k \varphi_k) {P}_k (I_n - K_k \varphi_k)^\top + W K_k Q K^\top_k).
\end{equation}
Afterwards, performing appropriate manipulations, and thanks to the symmetry of $P_k$ and $W$ is obtained
\begin{equation}
J_W = \text{Tr}(\lambda^{-1}W{P}_k - 2\lambda^{-1}W{P}_k\varphi^\top_k K^\top_k + \lambda^{-1}WK_k \Lambda K^\top_k),
\end{equation}
with $\Lambda = \varphi_k{P}_k\varphi^\top_k + \lambda Q$. Finally, determining $\frac{\partial J_W}{\partial K_k} = 0$ through Property \ref{pr:matrixdrevative} and solving for $K_k$ it results in $2\lambda^{-1}WP_k\varphi^\top_k = 2\lambda^{-1}WK_k\Lambda $ which yields to \eqref{eq:Correctiongain}\par 
Then, substituting \eqref{eq:Correctiongain} into \eqref{eq:Cov1} and do some manipulations, we have
\begin{align*}
P_{k+1} =\ & \lambda^{-1}{P}_k - \lambda^{-1}{P}_k \varphi^\top_k \Lambda^{-1} \varphi_k {P}_k \notag \\
& + \lambda^{-1}{P}_k \varphi^\top_k \Lambda^{-1} \Lambda \Lambda^{-1}  \varphi_k{P}^\top_k.
\end{align*}
Then, applying $\Lambda^{-1} \Lambda = I_p$, using the symmetric property of ${P}_k$ and making some manipulations, the following expression is achieved:
\begin{equation}
P_{k+1} = \lambda^{-1}(I_n - {P}_k \varphi^\top_k \Lambda^{-1} \varphi_k) {P}_k.
\end{equation}
Then replacing  ${P}_k \varphi_k^\top\Lambda^{-1}$ by $K^{\star}_{k}$ it yields to \eqref{eq:Cov2}.
\end{proof}
Through the Corollary \ref{co:Krank}, the optimal correction gain \eqref{eq:Correctiongain} satisfies the definition conditions in Proposition \ref{pro:rze}.
\begin{corollary}\label{co:Krank}
    Given PE condition of $\varphi_k$ as specified in Definition~\eqref{def:PE}, using correction gain \eqref{eq:Correctiongain} and covariance update matrix \eqref{eq:Cov2}, then $\operatorname{rank}(I_n - K_k \varphi_k)=n$.
\end{corollary}
\begin{proof}
    Given the covariance update matrix \ref{eq:Cov2}, we can write $(I_n - K^{\star}_{k} \varphi_k) = \lambda P_{k+1} {P}^{-1}_k$.
Since the $P_k$ is a positive definite matrix and it is of full rank under the PE condition, $(I_n - K^{\star}_{k} \varphi_k)$ is of full rank. 
\end{proof}
\begin{remark}
The $\lambda$ in \eqref{eq:R_update_a} scales up the generators to cover the uncertainties and serves as the forgetting factor that exponentially decreases the influence of past data in the covariance update in \eqref{eq:Cov2}. When there is no uncertainty in parameter $G_{\theta}=\textbf{0}$, choosing $\lambda=1$, which means the generators are updated only by the noise effect, and all data have the same weight in updating the covariance matrix. However, when $G_{\theta}\neq\textbf{0}$, the forgetting factor $\lambda$ is chosen less than $1$, which means the generators should be scaled up to cover these uncertainties in computing the boundary of the parameter set, prioritizing recent measurements in computing the covariance matrix.  
\end{remark}

\begin{remark}
The convergence of the updating center \eqref{eq:C_update} using the optimal gain \eqref{eq:Correctiongain} has been studied in \cite{bruggemann2021exponential}. It has been established that under the PE condition of $\varphi_k$, the parameter estimation error $\tilde{\theta}_k = \theta_k - \theta_{\text{tr},k}$ globally uniformly and ultimately converges to a bounded set.\par
Additionally, Proposition \ref{pro:rze} proved that for a sufficiently small $\lambda$, the updated zonotope in \eqref{eq:ZRLS_theta} can cover all deviations in time-varying parameters. However, if $\lambda$ is chosen too small, the covariance matrix update in \eqref{eq:Cov2} may cause the covariance matrix to diverge. Therefore, it is recommended in the literature to choose $\lambda$ in the range of $0.95$ to $1$ for normal recursive least squares with forgetting factor. With this selection, a dedicated method is required to estimate the system parameters, which are shown to be slow time-varying (e.g., with small $\delta_{\theta}$).
\end{remark}

\begin{remark}
    The most straightforward way to initialize the center in \eqref{eq:C_update} is to initialize the $C_0=\textbf{0}$ and initialize the initial covariance matrix in \eqref{eq:Cov2} with a large enough identity matrix $P_0=\tau I_n$ for $\tau\in \R_{>0}$ which means the initial center $C_0$ is not so accurate. On the other hand, in case of availability of collecting offline data, first collect at least $N_0$ snapshots (more precisely, enough snapshots so that $\Phi = [\varphi_0,\varphi_1,...,\varphi_{N_0}]^\top$ has rank $n+m$ and $Y = [y_1,y_2,...,y_{N_0}]^\top$), and then compute the initial $C_0=\Phi^\dagger Y$ and $P_0=\Phi^\dagger$ or $P_0=G_0G^\top_0$.
\end{remark}

\begin{remark}
Without loss of generality, to estimate the unknown parameters in \eqref{eq:regression}, we can use the vector-permutation form of this regression as
\begin{align}\label{eq:regressionvector}
\begin{split}
  \bar{y}_k &= \bar{\varphi}_k \bar{\theta}_{\mathrm{tr},k} + \bar{v}_k, \\
  \bar{\theta}_{\mathrm{tr},k} &= \bar{\theta}_{\mathrm{tr},k-1} + \delta \bar{\theta}_{k-1}.
  \end{split}
\end{align}
Where $\bar{y}_k=\text{vec}({y}_k)$ , $\bar{\theta}_{\mathrm{tr},k}=\text{vec}(\bar{\theta}_{\mathrm{tr},k})$, $\bar{v}_k=\text{vec}({v}_k)$, and $\bar{\varphi}_k=I_m\otimes \varphi_k$. Then all matrix zonotopes will be converted to the (vector) zonotope, and the center update equation in \eqref{eq:C_update} will be a vector $\bar{C}_k=\text{vec}(C_k)\in \R^{nm}$ and following the dimension of covariance matrix in \eqref{eq:Cov2} increased to $\bar{P}_k\in \R^{(nm\times nm)}$. Given the study in \cite{lai2024efficient} using vector-permutation makes more computational cost since the Kronecker product creates numerous zero terms in the regressor, increasing computational cost. In addition, using vector-permutation is beneficial when the noise distribution is different for each output, while in our proposed method, no distribution is considered for noise. 
   
\end{remark}
\begin{remark}
    Using the generator updating equation in \eqref{eq:R_update_b} can result in a large generator matrix after several time steps, potentially causing computational issues. To reduce this, one can employ the reduced-order generator $\downarrow_q (G_{\lambda_{k+1}})$ as described in Proposition \ref{cor:matrix-reduction}.
\end{remark}

%% file: Sections/5-RLSReachability.tex
\section{Data-Driven Reachability Analysis}\label{sec:RLSReachability}
In this section, we compute over-approximations of discrete-time LTV and nonlinear Lipschitz dynamic systems over a predefined time horizon, without assuming any prior knowledge of the underlying dynamic models. The available data are assumed to be corrupted by process noise. 
\subsection{Linear Time Varying Systems with Process Noise}\label{sec:reachlin_processnoise}
Consider $f(x_k,u_k)$ in \eqref{eq:sysnonlingen} is a discrete-time linear time-varying function. Then, the system \eqref{eq:sysnonlingen} can be written as follows
\begin{equation} \label{eq:sys}
\begin{aligned}
x_{k+1} &= A_{\text{tr},k} \, x_k + B_{\text{tr},k} \, u_k + w_k, \\
A_{\text{tr},k} &= A_{\text{tr},k-1} + \delta A_{k-1}, \\
B_{\text{tr},k} &= B_{\text{tr},k-1} + \delta B_{k-1}.
\end{aligned}
\end{equation}
where $A_{\text{tr},k} \in \mathbb{R}^{n_x \times n_x}$, and $B_{\text{tr},k} \in \mathbb{R}^{n_x \times n_u}$ are the true unknown time-varying system model, and $\delta A_{k} \in \mathbb{R}^{n_x \times n_x}$, and $\delta B_{k} \in \mathbb{R}^{n_x \times n_u}$ are change in the system model which $\bigg\|{\begin{bmatrix}  \delta A_{k} & \delta B_{k} \end{bmatrix}}\bigg\|_{max}\le \sigma_{AB} $ and the noise bounded by $\|w_k\|_{max}\leq\sigma_w$ for $k=0,1,...,\infty$.\par 

We define the unknown bounded noise set with Assumption \ref{ass:zon-noise}.
\begin{assumption}[\textbf{Noise Zonotope}]\label{ass:zon-noise}
The noise $w_k$ is contained in a known zonotope for all $k \in \mathbb{Z}_{\geq 0}$, 
i.e., $w_k \in \mathcal{Z}_w = \langle c_{\mathcal{Z}_w}, G_{\mathcal{Z}_w} \rangle$, 
where $c_{\mathcal{Z}_w} \in \mathbb{R}^{n_{\mathcal{Z}_w}}$ is the center, 
$G_{\mathcal{Z}_w} \in \mathbb{R}^{n_{\mathcal{Z}_w} \times \gamma_{\mathcal{Z}_w}}$ is the generator matrix, 
and the zonotope $\mathcal{Z}_w$ contains the origin.
\end{assumption}
To guarantee an over-approximation of the reachable sets for the unknown system in \eqref{eq:sys}, we estimate a time-varying set of models $\begin{bmatrix}A_k & B_k \end{bmatrix}$ which contains the true system model. By applying the Proposition \ref{pro:rze} and Theorem~\ref{th:optimalgain}, we estimate this time-varying set of models. Let $\hat{\mathcal{R}}_{k}$ denotes the reachable set computed based on the noisy data using matrix zonotopes. We propose Algorithm~\ref{alg:LTIreach1} to compute $\hat{\mathcal{R}}_{k}$ as an over-approximation of the exact reachable set $\mathcal{R}_{k}$. The following theorem proves that $\hat{\mathcal{R}}_{k} \supseteq \mathcal{R}_{k}$. 
\begin{theorem}[\scriptsize{\textbf{Reachability Analysis of the LTV Systems}}]
\label{th:reach_lin}
Given a minimum enough dense data on the set $D$ with length $N_D$ which the PE condition in Definition \ref{def:PE} is satisfied, and choosing a large enough $G_{{\lambda}_0}$ such that $[A_{\text{tr},0},B_{\text{tr},0}]\in \tilde{\mathcal{M}}_0=\zono{[A_0,B_0],G_{{\lambda}_0}}$ where $[A_0,B_0]$ is initial model estimation, then the reachable sets computed in Algorithm~\ref{alg:LTIreach1} over-approximates the exact reachable sets, i.e., $\hat{\mathcal{R}}_{k} \supseteq \mathcal{R}_{k}$ for $k=0,1,\dots N-1$.
\end{theorem}
\begin{proof}
Considering $w_k\in\mathcal{Z}_w$, we find $\mathcal{M}_k$ such that $[A_{\text{tr},k},B_{\text{tr},k}]\in \mathcal{M}_k$. Therefore we show for all $\begin{bmatrix}x_k\\u_k \end{bmatrix}\in \mathcal{F}$, we have
\begin{align*}
    x_{k+1} \in \mathcal{M}_k \begin{bmatrix} x_k \\ u_k \end{bmatrix} \oplus \mathcal{Z}_w \oplus \mathcal{Z}_{{\epsilon}}.
\end{align*}
for all $k = 0, 1, \dots,N-1$. Since the variations in the system model in~\eqref{eq:sys} are bounded by $\sigma_{AB}$, Assumption~\ref{ass:parametebound} guarantees that the perturbations satisfy $[\delta_{A_i}, \delta_{B_i}] \in G_{AB}$. By Proposition~\ref{pro:rze} and Theorem~\ref{th:optimalgain}, given $[A_{\text{tr},0}, B_{\text{tr},0}] \in \tilde{\mathcal{M}}_0$, there exists a scalar $\lambda \in (0,1]$ such that $[A_{\text{tr},i}, B_{\text{tr},i}] \in \tilde{\mathcal{M}}_i$ for all $i\ge0$. 

Since no new measured data is available when computing the reachable sets for steps \( k = 0, 1, \dots, N-1 \), we assume that the maximum deviation in the system model is bounded by \( \sigma_{AB} \). Applying Lemma~\ref{le:matrix_upper_correctedz_inf}, the zonotope
\[
    \mathcal{M}_k =  \zono{C_{\tilde{\mathcal{M}}_{i}},[G_{\tilde{\mathcal{M}}_{i}},G_{\hat{\mathcal{M}}_k}]} 
\]
where 
\[
    \hat{\mathcal{M}}_k = \zono{\textbf{0}, \left(k\sigma_{AB} E_1,\, k\sigma_{AB} E_2,\, \dots,\, k\sigma_{AB} E_{n_xn_u}\right)}
\]
contains all possible system models at time step \( k \).\par
To account for the spatial generalization error of the true system dynamics for any state-input pair $z_k=\begin{bmatrix}x^\top_k&u^\top_k\end{bmatrix}^\top \in \mathcal{F}$ at time step $k$. Since $\mathcal{F}$ is compact, we assume that all points $z_j \in D_{-}$ for $j=1,2,...,N_D$ are dense in $\mathcal{F}$. In other words, there exists some $\delta \geq 0$ such that for every $z_k \in \mathcal{F}$, there exists a $z_j \in D_{-}$ such that
\begin{align*}
\| {\mathcal{M}_k}z_k - {\mathcal{M}_k}z_j \| \leq I_{\mathcal{M}_{max}}  \delta.
\end{align*}
Where $I_{\mathcal{M}_k}=\text{int}(\mathcal{M}_k) $ and $I_{\mathcal{M}_{max}}=\sup_{k}\FFnorm{I_{\mathcal{M}_k}}$. This spatial generalization error is captured by the zonotope \cite{farjadnia2024robust}
\begin{align*}
 \mathcal{Z}_{{\epsilon}} = \zono{\textbf{0},I_{\mathcal{M}_{max}}  \delta/2,\dots,I_{\mathcal{M}_{max}} \delta/2)}.
\end{align*}
\end{proof}

\begin{algorithm}[t]
  \caption{LTV-Reachability}
  \label{alg:LTIreach1}
  \textbf{Input}: Input-state trajectory ${D_{-}}$, initial set $\mathcal{X}_{0}$, process noise zonotope $\mathcal{Z}_w$ and matrix zonotope $G_{AB}$, and input zonotope $\mathcal{U}_k$, $\forall k = 0, \dots,N-1$\\
    \textbf{Output}: Over approximate reachable sets $\hat{\mathcal{R}}_{k}, \forall k = 0, \dots,N-1$ 
  \begin{algorithmic}[1]
  \State $\hat{\mathcal{R}}_{0} =\mathcal{X}_{0}$,
    \State  Update the Center of model set,  \Comment{Use ~\ref{eq:C_update}} 
    \State  Update the generators of model set,  \Comment{Use ~\ref{eq:R_update}} 
    \State Update the identified model set $\tilde{\mathcal{M}}_i$, \Comment{Use ~\ref{eq:ZRLS_theta}}
    \State  Update convenience matrix,  \Comment{Use ~\ref{eq:Cov2}}

    \label{ln:alglipMtilde}
  \State Compute $\hat{\delta}$, \Comment{Use ~\ref{eq:deltaconstant}},
   \State $\mathcal{Z}_{{\epsilon}} = \zono{\textbf{0},I_{\mathcal{M}_{max}}  \hat{\delta}/2,\dots,I_{\mathcal{M}_{max}}  \hat{\delta}/2)}$, \label{ln:alglipZeps}

  \For{$k = 0:N-1$}
  \State $\hat{\mathcal{M}}_k = \zono{\textbf{0},[k\sigma_{AB} E_1,k\sigma_{AB} E_2,\dots,k\sigma_{AB} E_{n_xn_u}]}$,
  \State $        \mathcal{M}_k =  \zono{C_{\tilde{\mathcal{M}}_{i}},[G_{\tilde{\mathcal{M}}_{i}},G_{\hat{\mathcal{M}}_k}]} $
   
  \State $\hat{\mathcal{R}}_{k+1} =\mathcal{M}_k~ (\hat{\mathcal{R}}_{k} \times \mathcal{U}_{k}  ) + \mathcal{Z}_{{\epsilon}} +  \mathcal{Z}_w $. \label{ln:algLTIRhat}

  \EndFor

  \end{algorithmic}
\end{algorithm}
By compactness of $\mathcal{U}_k$, $\mathcal{R}_k$ for $k = 0,\dots,N$, the product space \(\mathcal{F} \) is compact. Therefore, for any $\delta > 0$, there exists a finite dataset $ \left\{z_j\right\}_{j=1}^{N_D} \subset \mathcal{F}, \quad \text{where} \quad z_j = \begin{bmatrix} (X_-)_{.,j}\\ (U_-)_{.,j} \end{bmatrix},$
that is $\delta$-dense in $\mathcal{F}$, i.e.,\(\forall z \in \mathcal{F}, \; \exists z_j \in D_{-} \text{ such that } \|z_k - z_j\| \leq \delta.\)
 \begin{remark}(\cite{alanwar2023data})
\label{rem:approxdelat}
Note that computing $\delta$ is non-trivial in practice. If we assume that the data is evenly spread out in the compact input set, then the following can be a good approximation of the $\delta$
\begin{align}
    \hat{\delta} &= \max_{z_i \in D_{-}} \min_{z_j \in D_{-}, j \neq i} \| z_i - z_j \|.
    \label{eq:deltaconstant}
\end{align}
The $\hat{\delta}$ may be estimated either from existing recorded data or from an updated sliding window $D_{-}$, which results in a more conservative estimate.

\end{remark}

\begin{remark}
Practical approaches exist for estimating $\sigma_{AB}$ in order to define $\hat{\mathcal{M}}_k$ in Algorithm~\eqref{alg:LTIreach1}. Moreover, the time-varying part of~\eqref{eq:sys} can be expressed as
\begin{align*}
    A_{\text{tr},k+N} &= A_{\text{tr},k} + \sum_{j=0}^{N-1} \delta A_{k+j}, \\
    B_{\text{tr},k+N} &= B_{\text{tr},k} + \sum_{j=0}^{N-1} \delta B_{k+j}.
\end{align*}
Without loss of generality, for slowly time-varying systems over a short horizon of $N$ steps, it can be assumed that 
$\sum_{j=0}^{N-1} \delta A_{k+j} \simeq 0$ and 
$\sum_{j=0}^{N-1} \delta B_{k+j} \simeq 0$. 
This approximation ensures that the over-approximation of the reachable sets remains valid. In fact, for most slowly varying linear systems, model variations become significant only over longer horizons, i.e., when $N \gg k$.
\end{remark}

\subsection{Lipschitz Nonlinear Systems} \label{sec:libs}

We consider a discrete-time Lipschitz nonlinear control system in \eqref{eq:sysnonlingen}. A local linearization of \eqref{eq:sysnonlingen} is performed by a Taylor series expansion around the time varying linearization point \cite{alanwar2023data} as $z_k^\star=\begin{bmatrix}x_k^\star\\u_k^\star \end{bmatrix}$
\begin{align*}
f(z_k) =& f(z_k^\star) + \frac{\partial f(z)}{\partial z_k}\Big|_{z_k=z_k^\star} (z_k - z_k^\star)+ \dots 
\end{align*}
The infinite Taylor series \cite{conf:taylor} can be represented by a first-order Taylor series and a Lagrange remainder term $L(z_k)$ \cite[p.65]{conf:thesisalthoff}, that depends on the model, as follows.
\begin{align}
    f(z_k) = f(z_k^\star) + \frac{\partial f(z_k)}{\partial z_k}\Big|_{z_k=z_k^\star} (z_k - z_k^\star)
+ L(z_k).
\label{eq:linfL}
\end{align}
Since the model is assumed to be unknown, we aim to over-approximate $L(z)$ from data. We rewrite \eqref{eq:linfL} as follows
\begin{align}
\begin{split}
f(x_k,u_k) ={}& f(x_k^\star,u_k^\star) - \tilde{A}_{\text{tr}} x_k^\star - \tilde{B}_{\text{tr}} u_k^\star \\
&\quad + \tilde{A}_{\text{tr}} x_k + \tilde{B}_{\text{tr}} u_k + L(x_k,u_k)
\end{split}
\label{eq:LinearTalor}
\end{align}
where 
\[\tilde{A}_{\text{tr},k} =  {\frac{\partial f(x_k,u_k)}{\partial x_k}\Big|_{x=x_k^\star,u=u_k^\star}},
\tilde{B}_{\text{tr},k} = {\frac{\partial f(x_k,u_k)}{\partial u_k}\Big|_{x_k=x_k^\star,u_k=u_k^\star}}.\]
By substituting $\Delta f^\star_{\mathrm{tr},k}
=f(x_k^\star,u_k^\star) - \tilde{A}_{\text{tr},k} x^\star -\tilde{B}_{\text{tr},k}u_k^\star$ in  \eqref{eq:LinearTalor}, we obtain:
\begin{align}
f(x_k,u_k) = \begin{bmatrix}\Delta f^\star_{\mathrm{tr},k}
 & \tilde{A}_{\text{tr},k} & \tilde{B}_{\text{tr},k}\end{bmatrix} \begin{bmatrix}\textbf{1}\\x_k\\ u_k\end{bmatrix}+L(x_k,u_k).
\label{eq:LinearTalor2}
\end{align}
where \(\Delta f^\star_{\mathrm{tr},k}\), \(\tilde{A}_{\text{tr},k}\) and  \(\tilde{B}_{\text{tr},k}\) are unknown linearized system model which we use EF-ZRLS to estimate time-varying sets which contains these matrices. 
To over-approximate the remainder term $L(z_k)$ from data, we need to assume that $f$ is Lipschitz continuous for all $z$ in the reachable set $\mathcal{F}$. The continuity of nonlinear systems can be defined using the Lipschitz condition.
\begin{assumption}
    It holds that $f: \mathcal{F} \rightarrow \mathbb{R}^{n_x}$ is Lipschitz continuous, i.e., that there is some $L^\star \geq 0$ such that 
    $\| f(z) - f(z^{\prime}) \|_2 \leq L^\star \| z - z^{\prime}\|_2$
    holds for all $z, z^{\prime} \in \mathcal{F}$.
    \label{as:lipschitz}
\end{assumption}

For data-driven methods applied to nonlinear systems, Lipschitz continuity is a common assumption (e.g., \cite{conf:montenbruckLipschitz,conf:novaraLipschitz}). 
Our goal is to approximate the Lipschitz nonlinear system by a LTV system over discrete time steps. 
By invoking Proposition~\ref{pro:rze} and Theorem~\ref{th:optimalgain} and applying them to~\eqref{eq:LinearTalor2} without the Lipschitz term $L(x,u)$, we obtain a time-varying set containing $\begin{bmatrix}\Delta f^\star_{\mathrm{tr},k} & \tilde{A}_{\text{tr},k} & \tilde{B}_{\text{tr},k}\end{bmatrix}$.
Then, by employing a sliding window of length $N_D$ that stores the most recent $N_D$ data points, we compute an over-approximation of $L(x_k,u_k)$. The following theorem proves the over-approximation of the reachable sets ${\mathcal{R}}^\prime_{k}$ out of Algorithm~\ref{alg:LipReachability} for the exact reachable sets $\mathcal{R}_k$.

\begin{theorem}[\scriptsize{Reachability Analysis of Lipschitz Systems}]
\label{th:reachdisnonlin}
Given sufficiently dense data on the set \( D \) of length \( N_D \), satisfying the PE condition in Definition~\ref{def:PE}, and choosing a large enough zonotope radius \( G_{{\lambda}_0} \) such that the linearization model 
\[
\begin{bmatrix}\Delta f^\star_{\mathrm{tr},0} & \tilde{A}_{\mathrm{tr},0} & \tilde{B}_{\mathrm{tr},0} \end{bmatrix} 
\in {\tilde{\mathcal{M}}}_{f,0} = \zono{
\begin{bmatrix} \Delta f^\star_{0} & \tilde{A}_{0} & \tilde{B}_{0} \end{bmatrix}, G_{\lambda_0} },
\]
where \( \begin{bmatrix} \Delta f^\star_{0} & \tilde{A}_{0} & \tilde{B}_{0} \end{bmatrix} \) is the initial model estimation,  
then the reachable sets computed in Algorithm~\ref{alg:LipReachability} over-approximates the exact reachable sets, i.e., $\mathcal{R}_{k} \subseteq \mathcal{R}^{\prime}_k, \quad \forall k = 0, 1, \dots, N-1.$
\end{theorem}

\begin{proof}
We use the proof of \cite[Theorem 5]{alanwar2023data} and extend it to this work. We know from~\eqref{eq:LinearTalor2} that
\begin{align*}
    f(z_k) = M_{\text{tr}_k} \begin{bmatrix} \textbf{1} \\ z_k \end{bmatrix} + L(z_k),
\end{align*}
where  $M_{\text{tr}_k} = \begin{bmatrix}\Delta f^\star_{\mathrm{tr},k}
 & \tilde{A}_{\text{tr},k} & \tilde{B}_{\text{tr},k}\end{bmatrix}$.
Hence, we need to show that $\mathcal{Z}_L + \mathcal{Z}_{\bar{\epsilon}}$  over-approximates the modeling mismatch and the term $L(z)$
\begin{align*}
   f(z_k) \in {\mathcal{M}}_k \begin{bmatrix} \textbf{1} \\ z_k \end{bmatrix} \oplus \mathcal{Z}_L \oplus \mathcal{Z}_{\bar{\epsilon}}.
\end{align*}
where $M_{\text{tr}_k}\in{\mathcal{M}}_{k}$ for all $z_k \in \mathcal{F}$. 
Since \( f(z_k) \) is Lipschitz, we can conclude that \( M_{\mathrm{tr},k} \) is bounded. Therefore, we can write 
\[
M_{\mathrm{tr},i} = M_{\mathrm{tr},i-1} + \delta M_{i-1}.
\]
where \( \delta M_{i-1} \) denotes the variation at step \( i-1 \). If we can show that \( \| \delta M_{i-1} \|_{max} \leq {\sigma}_{M} \), then, similar to Assumption \ref{ass:parametebound}, we can define a set \( G_{M} \) that over-approximates the system model changes, i.e., \( \delta M_{i-1} \in G_{M} \).

Given that \( w_i \in \mathcal{Z}_w \) and by invoking Proposition~\ref{pro:rze} and Theorem \ref{th:optimalgain}, we can find a \( \lambda \in (0,1] \) such that 
\[
M_{\mathrm{tr},i} \in \tilde{\mathcal{M}}_{f,i}, \quad \forall i\ge0.
\]
while \( M_{\mathrm{tr},0} \in \tilde{\mathcal{M}}_{f,0} \). By defining $\bar{i}=i-N_D+j$To over-approximate $L(z_j)$, we know that for all $z_j\in D_{-}$, we can write 
\begin{align*}
    (X_+)_{\cdot,j} - (W_-)_{\cdot, j} = (C_{\tilde{\mathcal{M}}_{f,\bar{i}}} + \Delta M_{\bar{i}}) \begin{bmatrix} \textbf{1} \\z_j  \end{bmatrix} \\ \nonumber
     + L(z_j),\forall j =1,2,\dots,N_D.
\end{align*}
where $\Delta M_{\bar{i}}=M_{\mathrm{tr},\bar{i}}-C_{\tilde{\mathcal{M}}_{f,\bar{i}}}$. This implies 
\begin{align}
     (X_+)_{\cdot,j}  - C_{\tilde{\mathcal{M}}_{f,\bar{i}}} \begin{bmatrix} \textbf{1}\\ z_j  \end{bmatrix} \in
     \Delta\tilde{\mathcal{M}}_{f,\bar{i}}  + L(z_j)\\ \nonumber
     + \mathcal{Z}_w.
    \label{eq:deltamindata}
\end{align}
where $\Delta\tilde{\mathcal{M}}_{f,\bar{i}}=\zono{\textbf{0},G_{\tilde{\mathcal{M}}_{f,\bar{i}}}}$. Next, we aim to find one zonotope $\mathcal{Z}_L$ that over-approximates $ L(z_j) $ for all the data points, i.e.,  $\forall z_j \in D_{-}$
\begin{align*}
 (X_+)_{\cdot,j} - C_{\tilde{\mathcal{M}}_{f,\bar{i}}} \begin{bmatrix} \textbf{1} \\z_j  \end{bmatrix}
 \in \Delta\tilde{\mathcal{M}}_{f,\bar{i}} \oplus \mathcal{Z}_L \oplus \mathcal{Z}_w.    
\end{align*}
 This can be found by the upper bound ($\overline{l}$ in line~\ref{ln:algliplupper}) and lower bound ($\underline{l}$ in line~\ref{ln:alglipllower}) from data, and hence $\mathcal{Z}_L$ in line~\ref{ln:alglipZl} in Algorithm \ref{alg:LipReachability}. Thus, we can over-approximate the model mismatch and the nonlinearity term for all data points $z_j \in D_{-}$, by
\begin{align*}
    f(z_j) \in \tilde{\mathcal{M}}_{f,\bar{i}} \begin{bmatrix} \textbf{1} \\ z_j  \end{bmatrix} + \mathcal{Z}_L.
\end{align*}
Given the covering radius $\delta$ of our system together with Assumption~\ref{as:lipschitz}, we know that for every $z \in \mathcal{F}$, there exists a $z_j \in D_{-}$ such that
    $\| f(z_k) - f(z_j) \| \leq L^\star \| z_k - z_j \| \leq L^\star \delta$.
This yields $\mathcal{Z}_{\bar{\epsilon}} = \zono{\textbf{0},\text{diag}(L^{\star^{(1)}} \delta/2,\dots,L^{\star^{(n_x)}} \delta/2)}$.\par
To complete the proof, similar to Theorem~\ref{th:reach_lin}, we need to use $\hat{\mathcal{M}}_k$ to overapproximate the reachable sets when no new measurement is available by using Lemma~\ref{le:matrix_upper_correctedz_inf}. Then using  
\[
    \mathcal{M}_{f,k} = \zono{C_{\tilde{\mathcal{M}}_{f,i}},[G_{\tilde{\mathcal{M}}_{f,i}},G_{\hat{\mathcal{M}}_k}]}
\]
yields 
\begin{align*}
    f(z_k) \in \mathcal{M}_{f,k} \begin{bmatrix} \textbf{1} \\ z_k \end{bmatrix} +  \mathcal{Z}_L + \mathcal{Z}_{\bar{\epsilon}}.
\end{align*}
\end{proof}
For an infinite amount of data, i.e., $\delta \rightarrow 0$, we can see that $\mathcal{Z}_{\bar{\epsilon}} \rightarrow 0$, i.e., the formal $\mathcal{Z}_L$ then fully captures the modeling mismatch and the Lagrange remainder. 
\begin{algorithm}[t]
  \caption{Nonlinear Lipschitz-Reachability}
  \label{alg:LipReachability}
  \textbf{Input}: Input-state trajectory ${D}$, initial set $\mathcal{X}_{0}$, process noise zonotope $\mathcal{Z}_w$ and matrix zonotope $G_{M}$, and input zonotope $\mathcal{U}_k$. \\
  \textbf{Output}: Over approximate reachable sets ${\mathcal{R}}^\prime_{k}, \forall k = 0, \dots,N-1$ 
  \begin{algorithmic}[1]
  \State ${\mathcal{R}}^\prime_{0} =\mathcal{X}_{0}$,

    \State  Update the Center of model set,\Comment{Use ~\ref{eq:C_update}} 
    \State  Update the generators of model set, \Comment{Use ~\ref{eq:R_update}} 
    \State Update the identified model set, $\tilde{\mathcal{M}}_{f,i}$ \Comment{Use ~\ref{eq:ZRLS_theta}}
    \State  Update Convenience Matrix,  \Comment{Use ~\ref{eq:Cov2}} 
\State $\overline{l} = \max_j \Big( (X_{+})_{.,j} - 
C_{\tilde{\mathcal{M}}_{f,\bar{i}}}
\begin{bmatrix} \textbf{1} \\ (X_{-})_{.,j} \\ (U_{-})_{.,j} \end{bmatrix} \Big)$
\label{ln:algliplupper}

\State $\underline{l} = \min_j \Big( (X_{+})_{.,j} - 
C_{\tilde{\mathcal{M}}_{f,\bar{i}}}
\begin{bmatrix} \textbf{1} \\ (X_{-})_{.,j} \\ (U_{-})_{.,j} \end{bmatrix} \Big)$
\label{ln:alglipllower}

\State $\mathcal{Z}_{L} = \text{zonotope}(\underline{l},\overline{l}),$ \label{ln:alglipZl} 
\State Compute $\hat{L}^\star$ and $\hat{\delta}$, \Comment{Use ~\ref{eq:lipschizconstant} and \ref{eq:deltaconstant}}
  \State $\mathcal{Z}_{\bar{\epsilon}} = \zono{\textbf{0},\textup{diag}(\hat{L}^{\star^{(1)}} \hat{\delta}/2,\dots,\hat{L}^{\star^{(n_x)}} \hat{\delta}/2)}$, \label{ln:alglipZeps}
  \For{$k = 0:N-1$}
   \State $\hat{\mathcal{M}}_k = \zono{\textbf{0},[k\sigma_{AB} E_1,k\sigma_{AB} E_2,\dots,k\sigma_{AB} E_{n_xn_u}]}$,
  \State $    \mathcal{M}_{f,k} =  \zono{C_{\tilde{\mathcal{M}}_{f,i}},[G_{\tilde{\mathcal{M}}_{f,i}},G_{\hat{\mathcal{M}}_k}]}$,
  \State ${\mathcal{R}}^\prime_{k+1} {=} \mathcal{M}_{f,k} \Big(1 \times{\mathcal{R}}^\prime_{k}\times \mathcal{U}_k\Big) +  \mathcal{Z}_L + \mathcal{Z}_{\bar{\epsilon}}  +  \mathcal{Z}_w$\label{ln:alglipRprime}.
  \EndFor

  \end{algorithmic}
\end{algorithm}
\begin{remark}[~\cite{alanwar2023data}]
\label{rem:approx}
Note that determining $L^\star$ is non-trivial in practic. If we assume that the data is evenly spread out in the compact input set of $f$, then the following can be a good approximation of the upper bound on $L^\star$ for each dimension $o$:
\begin{align}
     \hat{L}^{\star^{(o)}} &= \max_{z_i, z_j \in D_{-}, i\neq j}  \frac{\| f^{(o)}(z_i) - f^{(o)}(z_j) \|}{\| z_i - z_j \|}.
    \label{eq:lipschizconstant}
\end{align}
Computing the Lipschitz constant for each dimension decreases the conservatism, especially when the data has a different scale for each dimension.
Furthermore, note that from the proof of the above theorem, we see that for every reachability step, $k=1, \dots, N$, local information on $L^\star$ in the set $\mathcal{R}_k \times \mathcal{U}_k $ can be used. Since in each time step a new sample date is added to $D_{-}$, the $L^\star$ approximation is updated.
\end{remark}

\begin{remark}
    Initially, when the dataset $D$ is sparse, the over-approximation guarantees are not theoretically ensured. However, as data accumulates and $D$ becomes denser in $\mathcal{F}$, the reachability results become provably conservative.
\end{remark}

%% file: Sections/6-evaluation.tex
\section{Evaluation}\label{sec:evauation}
In this section, we apply the computational approaches to over-approximate the exact reachable sets from data. We collect online data, where the system model is unknown. We compare results with the method proposed in \cite{alanwar2023data} to compute the reachable sets, which uses a batch Least Squares (LS) method to estimate a set of models given a rich enough offline data set to estimate models. 

\subsection{Example 1: A LTV System}\label{ex1}
To demonstrate the usefulness of the presented approach, we consider the reachability analysis of a five-dimensional system $A_k\in \R^{5\times 5}$ and $B_k\in \R^{5}$ with a single input, which is a discretization of the system used in \cite[p.39]{conf:thesisalthoff} with sampling time $0.1$ sec. 
The initial set is chosen to be $\mathcal{X}_0=\zono{\textbf{1},0.1 I_{5}}$. The input set is $\mathcal{U}_k=\zono{10,2.25}$. We consider computing the reachable set when there is a random noise sampled from the zonotope  $\mathcal{Z}_w=\zono{\textbf{0},0.005I_{5}}$.

To evaluate the Algorithm \ref{alg:LTIreach1}, we simulate the following senarious,
\begin{enumerate}


    \item \label{ex:1.B}A trajectory of length $N_{{D}}=30$ is used as input data on ${D}$, which is rich enough to estimate parameters with the LS method. There is no change in parameters (e.g. $\sigma_{AB}=0$).  

    \item\label{ex:1.C} A trajectory of length $N_{{D}}=30$ is used as input data on ${D}$, which is rich enough to estimate parameters. The parameters change over time-steps as $\delta A_k=-0.0001 \textbf{1}_{5\times5}$ and $\delta B_k=0.0003 \textbf{1}_{5\times5}$ and $\sigma_{AB}=0.0003$.  
\end{enumerate}
For both scenarios, we consider the initial center $C_0=\textbf{0}$ and the initial convenience matrix $P_0 = 10^{7}I_{5}$, and the initial generator $G_0=1.5E_i$ for $i=1,2,...,n_xn_u$ where $E_i$ are the standard basis matrices with a single non-zero element of 1.\par

\begin{figure*}[!htbp]
\vspace{-0.05cm}
    \centering
    \begin{subfigure}[h]{0.32\textwidth}
     \centering
        \includegraphics[scale=0.3]{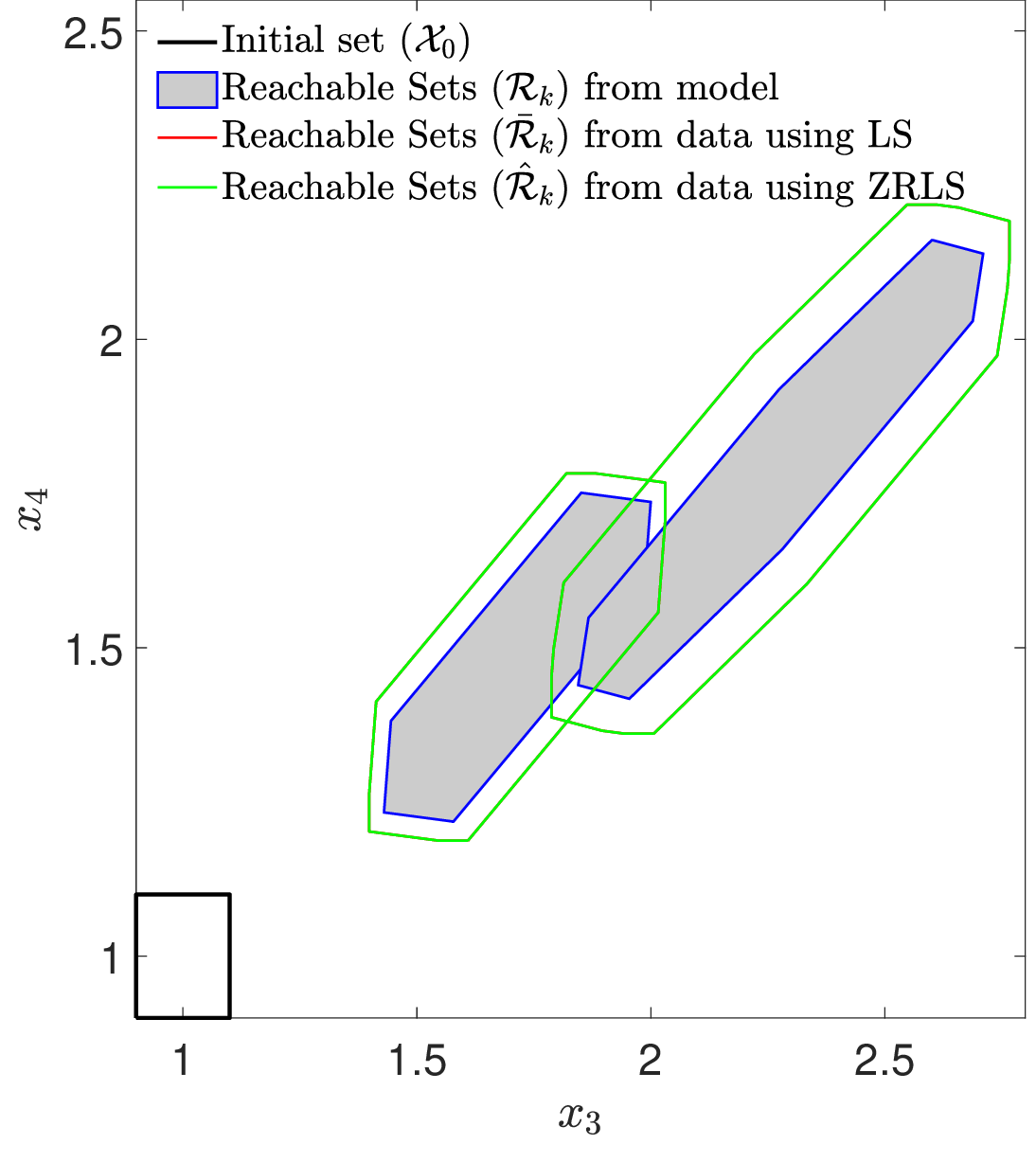}
        \caption{}
        \label{fig:ex1_B2}
    \end{subfigure}
    \begin{subfigure}[h]{0.32\textwidth}
     \centering
        \includegraphics[scale=0.3]{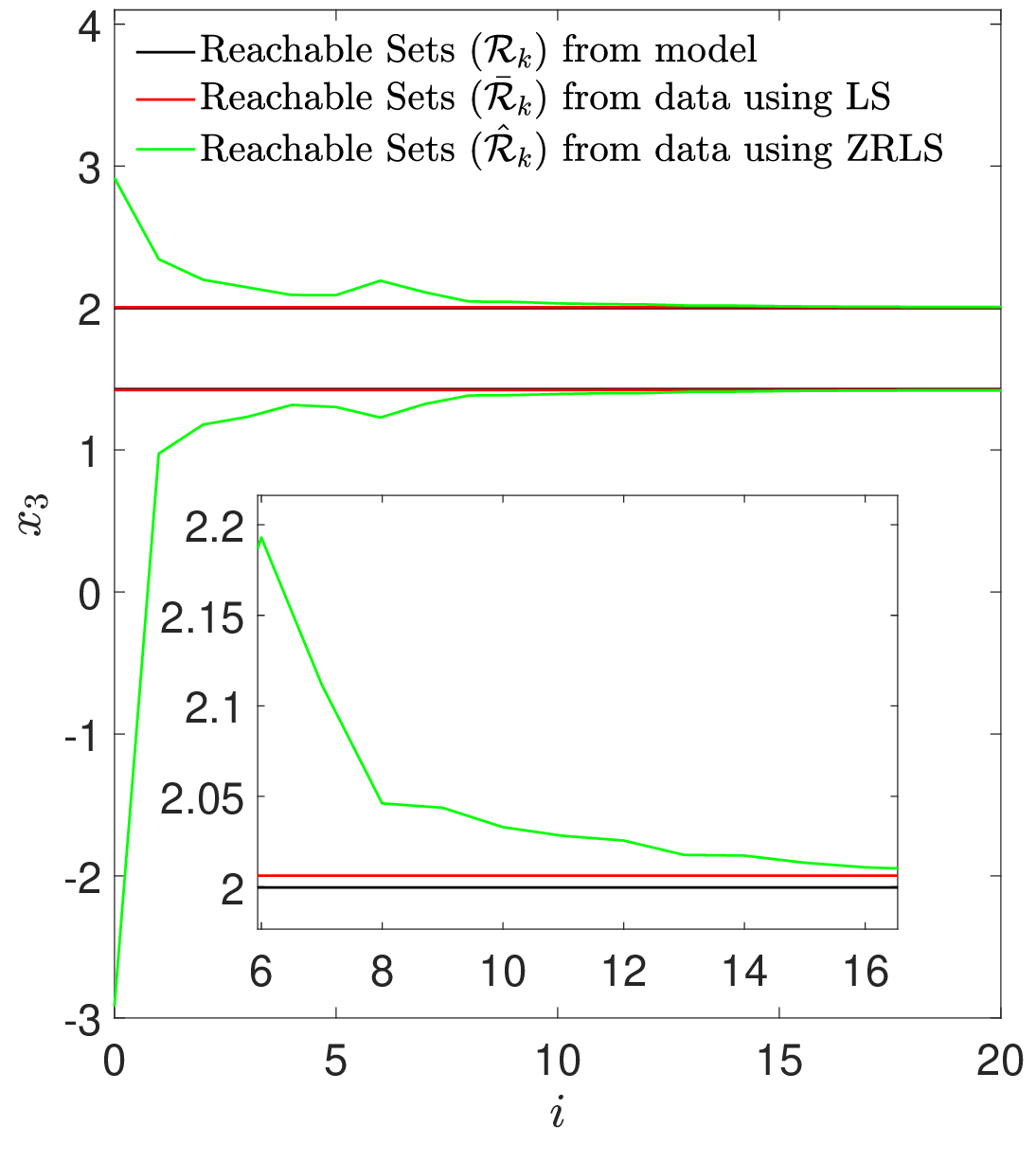}
        \caption{}
        \label{fig:ex1_A3}
    \end{subfigure}
    \begin{subfigure}[h]{0.32\textwidth}
     \centering
        \includegraphics[scale=0.3]{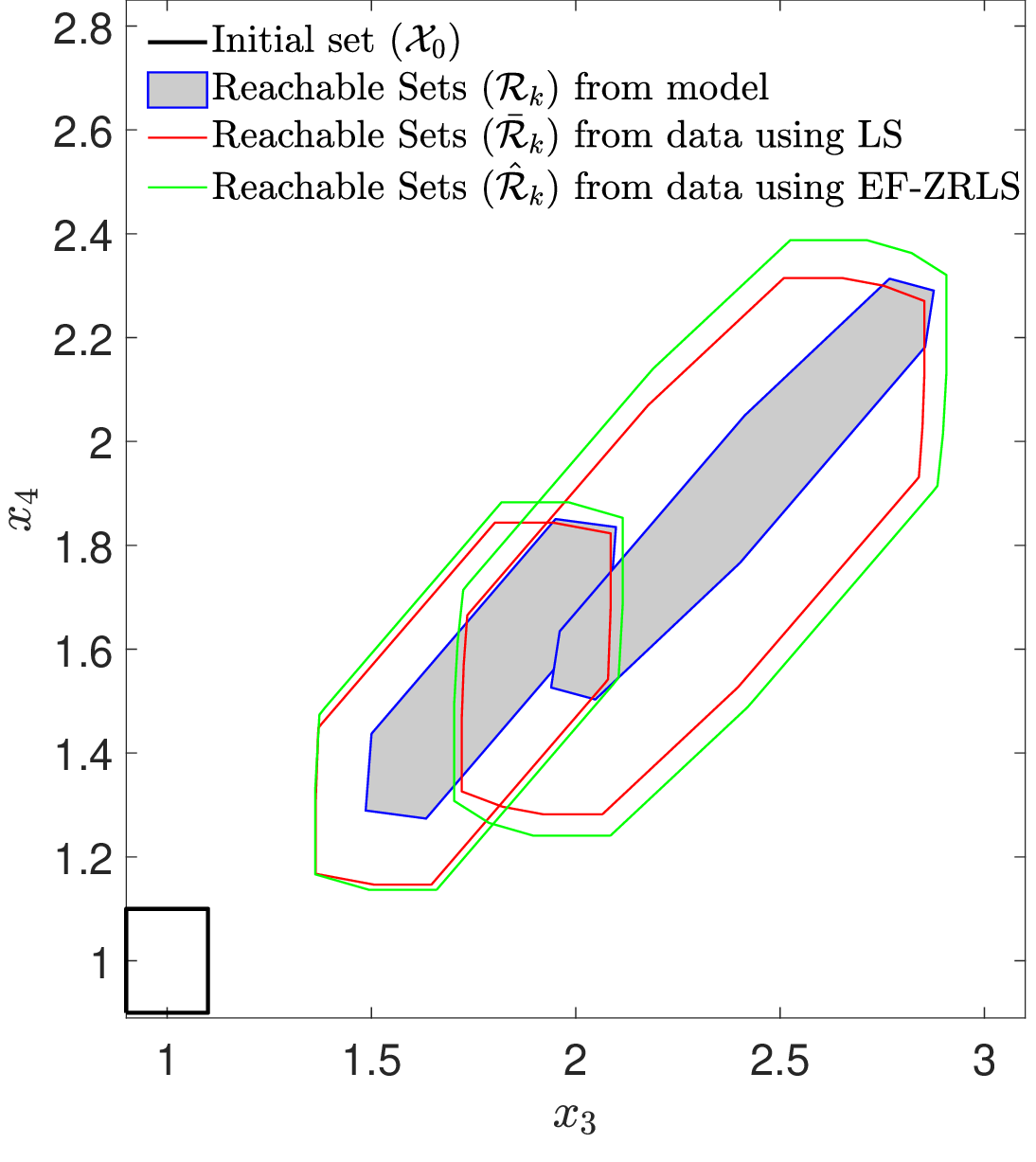}
        \caption{}
        \label{fig:ex1_C2}
    \end{subfigure}
\caption{The projection of the reachable sets of the LTV system in \eqref{eq:sys} computed via Algorithm~\ref{alg:LTIreach1} ($\hat{\mathcal{R}}_k$), the LS method ($\bar{\mathcal{R}}_k$) computed in \cite[Algorithm 1]{alanwar2023data}, from noisy input-state data. There is no change in system matrices (e.g. $\sigma_{AB}=0$)in (a) and (b), and $\sigma_{AB}=0.0003$ in (c).  }
    \label{fig:contreach}
     \vspace{-1mm}
\end{figure*}

The reachable set computed by Algorithm \ref{alg:LTIreach1} converges, after a number of iterations, to the reachable set obtained using the LS method \cite[Algorithm 1]{alanwar2023data}, as illustrated in Figure \ref{fig:ex1_B2}. In fact, the model set estimated by the proposed ZRLS (equivalent to EF-ZRLS with $\lambda=1$) eventually converges to the optimal model set computed by the LS method, which relies on an offline batch of data stored in $D$. To ensure an over-approximation of the reachable sets at the initial step, it is necessary to choose a sufficiently large initial generator so that the true system model is guaranteed to lie within the initial model set. Once this condition is satisfied, the estimated model set converges to the optimal set. The convergence rate is affected by the choice of the initial covariance matrix, the generator matrices, and the lower bound of the PE condition specified in Definition \ref{def:PE}. Figure \ref{fig:ex1_A3} illustrates the convergence of the over-approximated $x_2$ trajectory over time steps.

As shown in Figure \ref{fig:ex1_C2}, when the system model changes over time steps, the reachable sets computed by the LS method fail to over-approximate the exact reachable sets, while the reachable sets computed by Algorithm \ref{alg:LTIreach1} successfully over-approximate exact reachable sets. The EF-ZRLS with $\lambda=0.92$ , tries to minimize the covariance matrix affected by recent data, while the LS method considers the same weight of estimation error for all data on ${D}$. Decreasing $\lambda$ leads to having more conservative reachable states since the estimated model set is enlarged.

\subsection{Example 2: Lipschitz Nonlinear Systems}\label{ex2}

We apply the proposed data-driven reachability analysis to a continuous stirred tank reactor (CSTR) simulation model \cite{conf:nonlinearexample}. The initial set is a zonotope $\mathcal{X}_0 =\zono{[1.35,10.9]^\top,\text{diag}([ 0.005  , 0.3 ])}$. The input set $\mathcal{U}_k =\zono{[1.1 , -1.3]^\top,\text{diag}([ 0.1 , 0.2])}$ and the noise set $\mathcal{Z}_w=\zono{0,0.003I_{2}}$. We define a set $\bar{D}$ that contains 35 trajectories with length $N_{\bar{D}}$ collected offline to be used in the LS method in \cite{alanwar2023data} to compute reachability sets. The length of the moving window $D$ is $N_D=5$, which contains the latest data. The $\lambda=0.96$ is selected in EF-ZRLS, and the initial covariance matrix, initial set, and generators are the same as example 1. In addition, model variations are assumed to be negligible in computing the four short horizons $N=4$.  

In Figure \ref{fig:ex2_A}, the reachable sets computed by Algorithm \ref{alg:LipReachability} are shown over four steps ahead. As illustrated, these reachable sets are less conservative than those obtained using the LS method proposed in \cite{alanwar2023data}, since the Lagrangian remainder that captures nonlinearities is computed globally over the entire dataset $\bar{D}$, and all data in $\bar{D}$ are used to estimate the model for system linearization around the current point. In contrast, Algorithm \ref{alg:LipReachability} computes reachable sets in a more local manner: the linearized system model is updated recursively based on the latest data, and the set $\mathcal{Z}_L$, which over-approximates the Lagrangian remainder, is constructed from a moving window of data $D$. Increasing the length of stored data in $D$ results in more conservative reachable sets, while also enabling the capture of nonlinearities over a longer prediction horizon $N$. It is worth noting that the data corresponding to the prediction horizon are already contained in $\bar{D}$.

\subsection{Example 3: Autonomous Vehicle}

We utilized the JetRacer ROS AI Kit, an autonomous racing robot with an NVIDIA Jetson Nano Developer Kit (4GB RAM) and Raspberry Pi RP2040 dual-core microcontroller. 
The inputs to the vehicle are the linear and angular velocity, and the outputs are the position and heading of the vehicle.\par
We consider process noise bounded by the zonotope $\mathcal{Z}_w=\zono{\textbf{0}_{3\times3},0.0005I_{3}}$. The forgetting factor $\lambda=0.99$ with horizon step $N=3$ and $N_D=5$ is considered for implementing Algorithm \ref{alg:LipReachability}. The initial covariance matrix is set $P_0=100 I_{3}$, and the initial generator and center are the same as the previous examples. We consider two experimental scenarios. In the first scenario, the car and its environment remain unchanged over time. In the second one, a 2 $kg$ payload is added on top of the car.\par   
The over-approximated reachable sets computed using Algorithm~\ref{alg:LipReachability} are illustrated in Figure~\ref{fig:carA}. The larger sets observed at the beginning of the experiment result from the use of a large initial covariance matrix and initial generator, reflecting the uncertain initial value of the center of the model set due to the absence of the true initial model. In addition, our proposed method is less conservative compared to the LS method in \cite{alanwar2023data}, while we do not need to have pre-collected data for identifying model.\par 
In the second scenario, a payload is added on top of the car to evaluate the performance of our method under model changes. As shown in Figure~\ref{fig:carB}, the proposed method remains robust to model variations, whereas the LS method fails to over-approximate the reachable sets.

\begin{figure*}[!htbp]
\vspace{-0.05cm}
    \centering
    \begin{subfigure}[h]{0.32\textwidth}
     \centering
        \includegraphics[scale=0.26]{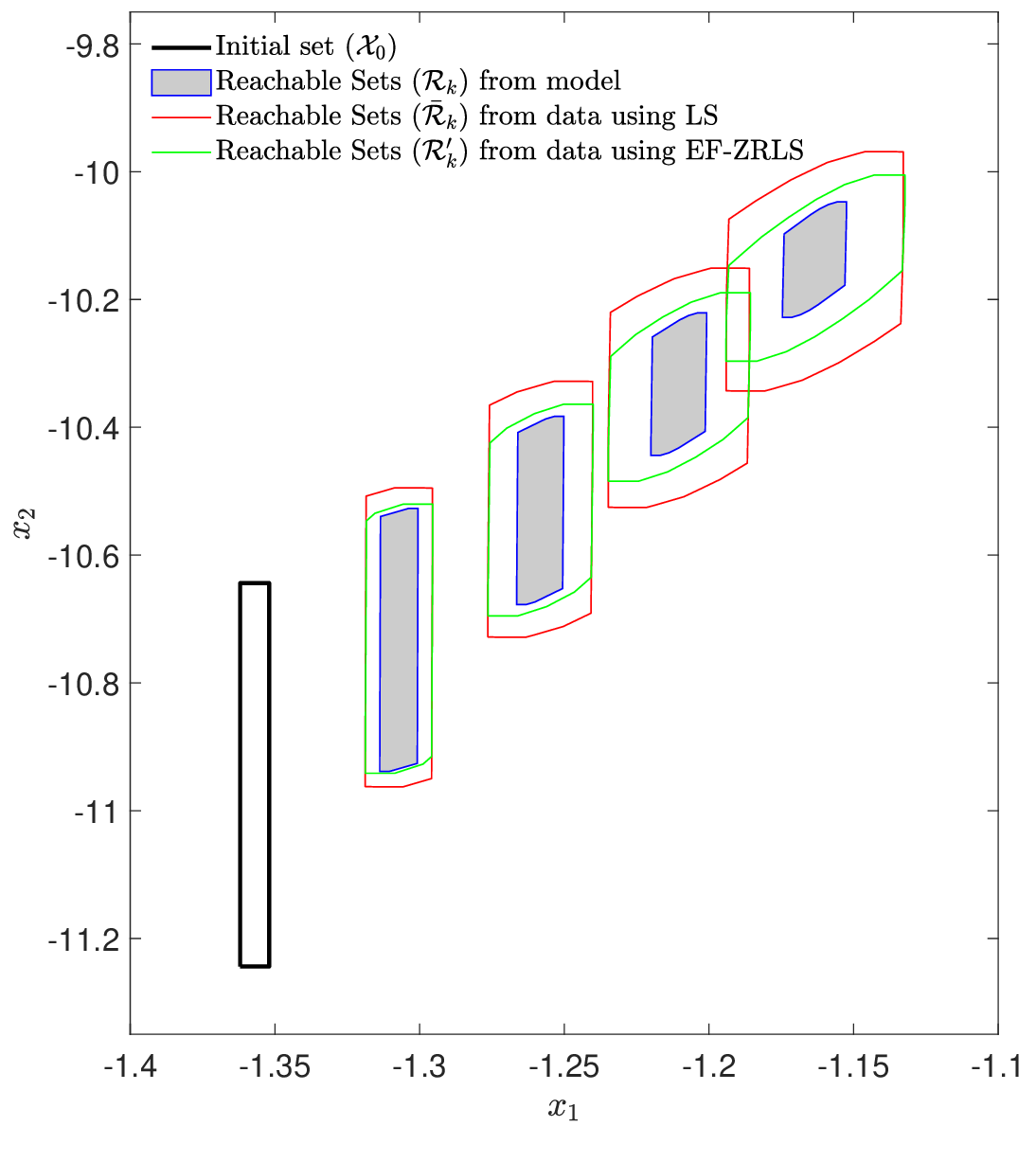}
        \caption{}
        \label{fig:ex2_A}
    \end{subfigure}
    \begin{subfigure}[h]{0.32\textwidth}
     \centering
        \includegraphics[scale=0.26]{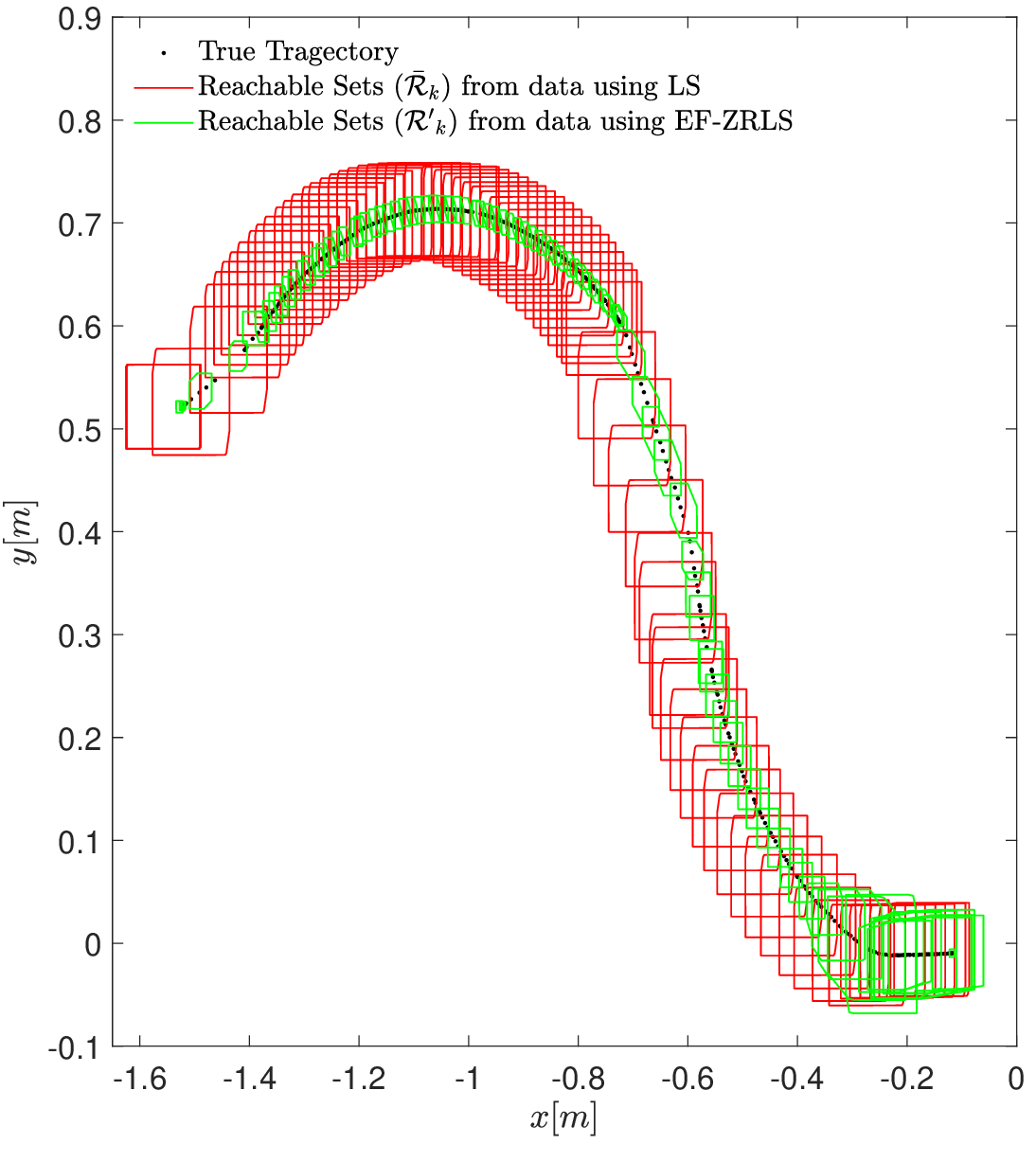}
        \caption{}
        \label{fig:carA}
    \end{subfigure}
    \begin{subfigure}[h]{0.32\textwidth}
     \centering
        \includegraphics[scale=0.26]{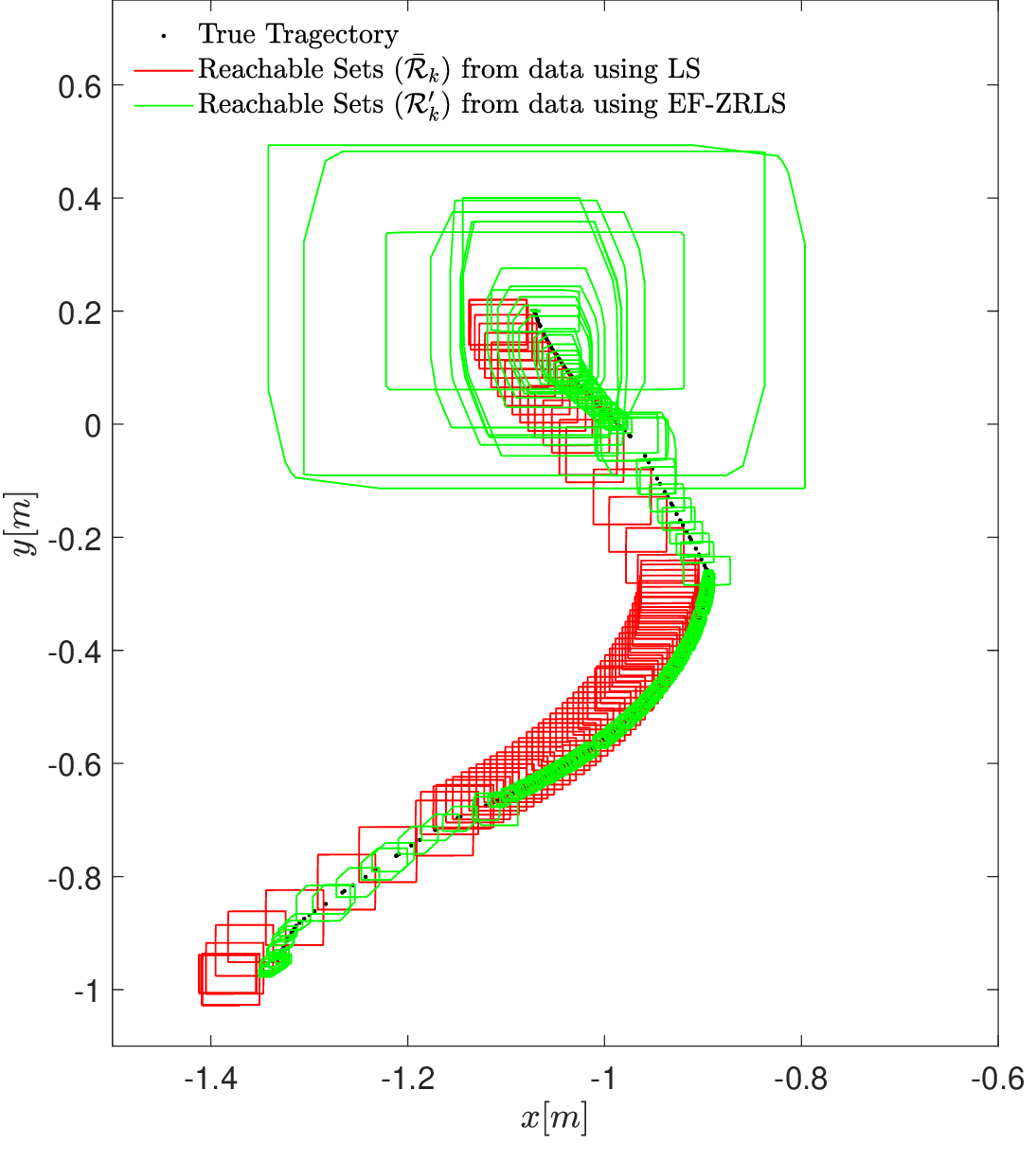}
        \caption{}
        \label{fig:carB}
    \end{subfigure}
\caption{The projection of the reachable sets of the Lipschitz nonlinear system in \eqref{eq:sysnonlingen} computed via Algorithm~\ref{alg:LipReachability} (${\mathcal{R}^{\prime}}_k$),  the LS method ($\bar{\mathcal{R}}_k$) in \cite[Algorithm 6]{alanwar2023data}, from noisy input-state data. (a): Nonlinear system of Example 2. (b): Real experiment of autonomous vehicle without payload shown in Figure \ref{fig:jetracer_nopay}. (c): Real experiment of autonomous vehicle with payload shown in Figure \ref{fig:jetracer_payload}.}
    \label{fig:contreach}
     \vspace{-1mm}
\end{figure*}

\begin{figure}[htbp]
    \centering
    \begin{subfigure}[b]{0.2\textwidth}
        \centering
        \includegraphics[width=\textwidth]{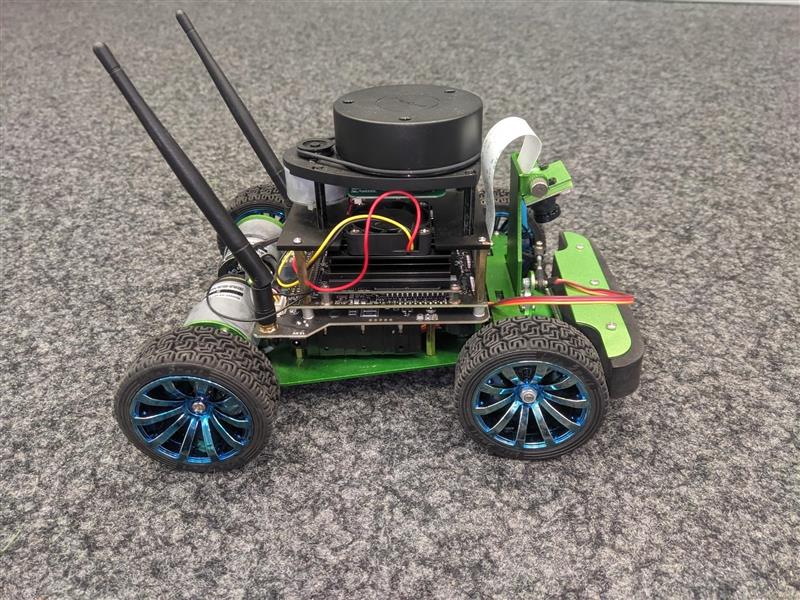}
         \caption{}
        \label{fig:jetracer_nopay}

    \end{subfigure}
    \begin{subfigure}[b]{0.2\textwidth}
        \centering
        \includegraphics[width=\textwidth]{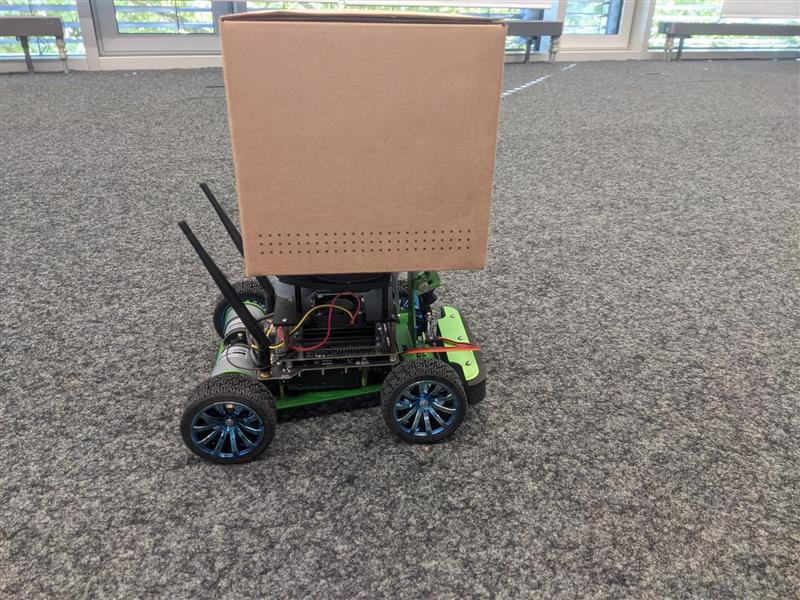}
                \caption{}
        \label{fig:jetracer_payload}

    \end{subfigure}
    \caption{The JetRacer ROS AI Kit. (a): without payload. (b): with 2 $kg$ payload.}
    \label{fig:jetracer}
         \vspace{-1mm}
\end{figure}

%% file: Sections/7-con.tex
\section{Conclusion}\label{sec:con}
We consider the problem of computing the reachable set using noisy data for two classes of time-discrete systems without requiring a mathematical model of the system or any pre-collected data. To do this, we introduced EF-ZRLS, which is an optimal set-based recursive estimation method to estimate a time-varying set of models given noisy input-state data for the multi-output regression models. Then, we over-approximated the reachable sets of LTV and Lipschitz nonlinear systems, which are a superset of the exact reachable sets.\par
In future works, we are going to extend the proposed EF-ZRLS to use adaptive $\lambda$ in \cite{bruggemann2021exponential} to decrease the conservatism of the estimating set of models.


\balance

%% file: Sections/8-Acknowledgement.tex
\section{Acknowledgement}\label{sec:acknowledgement}
This work is funded by the Deutsche Forschungsgemeinschaft (DFG, German Research Foundation) – Project number: 564740977.

%% file: Sections/Appendix_A.tex
\section*{Appendix A: Definitions and Proposition}
\begin{definition}[\textbf{Matrix Zonotope Covariance}~\cite{samada2023zonotopic}]\label{de:CovMatrixZono}
    Given the matrix zonotope of the Definition \ref{def:matzonotopes}, the convenience of a matrix zonotope is defined as
\begin{align*}
    \operatorname{cov}(\mathcal{M})&=\sum_{j=1}^{\gamma_{\mathcal{M}}}\tilde{G}^{(j)}_{\mathcal{M}}\tilde{G}^{(j)^\top}_{\mathcal{M}}=\tilde{G}_{\mathcal{M}}\tilde{G}^{\top}_{\mathcal{M}}.
\end{align*}
    
\end{definition}

\begin{definition}[\textbf{Interval Hull} \cite{samada2023zonotopic}]  
Given a zonotope of Definition \ref{def:zonotopes}, the \emph{interval hull} \( b(G_{\mathcal{Z}}) \subset \mathbb{R}^{n_x \times n_x} \) is defined as the smallest axis-aligned box containing \(\mathcal{Z}\), i.e.,  
\[
\mathcal{Z} \subseteq \zono{c_{\mathcal{Z}}, b(G_{\mathcal{Z}})}.
\]
where \( b(G_{\mathcal{Z}}) \) is a diagonal matrix with entries  
\[
b(G_{\mathcal{Z}})_{ii} = \sum_{j=1}^{\gamma_{\mathcal{Z}}} \left| \bigl(G_{\mathcal{Z}}\bigr)_{ij} \right|.
\]  
\end{definition}

\begin{definition}[\textbf{Zonotope Reduction Operator} \cite{samada2023zonotopic}] \label{def:reduction}
Given a zonotope of Definition \ref{def:zonotopes}, the \emph{reduction operator}
\[
\downarrow_q : \mathbb{R}^{n_{\mathcal{Z}} \times \gamma_{\mathcal{Z}}} \to \mathbb{R}^{n_x \times q}
\]
with \( q < \gamma_{\mathcal{Z}} \), produces a reduced generator matrix such that the resulting reduced zonotope satisfies
\[
\zono{c_{\mathcal{Z}}, G_{\mathcal{Z}}} \subseteq \zono{c_{\mathcal{Z}}, \downarrow_q(G_{\mathcal{Z}})}.
\]
\end{definition}

\begin{proposition}
[\textbf{Matrix Zonotope Reduction}]\label{cor:matrix-reduction}
Let a matrix zonotope $\mathcal{M}$ of Definition \ref{def:matzonotopes} and the zonotope $\mathcal{Z}$ of Definition \ref{def:zonotopes} to be the corresponding equivalent vector zonotope, where $G_{\mathcal{Z}} = \begin{bmatrix}\operatorname{vec}(G^{(1)}_{\mathcal{M}}), \dots, \operatorname{vec}(G^{\gamma_{\mathcal{M}}}_{\mathcal{M}})\end{bmatrix}$. If $\downarrow_q$ denotes the zonotope order-reduction operator from Definition \ref{def:reduction}, then a reduced generator matrix for $\mathcal{M}$ is given by
\[
\downarrow_q(\tilde{G}_{\mathcal{M}}) = \operatorname{unvec}(\downarrow_q(G_{\mathcal{Z}})),
\]
and the resulting reduced matrix zonotope satisfies
\[
\mathcal{M} \subseteq \zono{C_{\mathcal{M}},\downarrow_q(\tilde{G}_{\mathcal{M}})}.
\]
\end{proposition}
\begin{proof}
Let $\mathcal{Z} = \zono{c_{\mathcal{Z}}, G_{\mathcal{Z}}}$ be the vector zonotope equivalent to $\mathcal{M}$. From Definition \ref{def:reduction}, we know that applying the reduction operator to the generator matrix $G_{\mathcal{Z}}$ results in the containment relation:
\[
\zono{c_{\mathcal{Z}}, G_{\mathcal{Z}}} \subseteq \zono{c_{\mathcal{Z}}, \downarrow_q(G_{\mathcal{Z}})}.
\]
We can apply the linear and bijective map $\operatorname{unvec}(\cdot)$ to both sides, which preserves the containment relation:
\[
\operatorname{unvec}(\zono{c_{\mathcal{Z}}, G_{\mathcal{Z}}}) \subseteq \operatorname{unvec}(\zono{c_{\mathcal{Z}}, \downarrow_q(G_{\mathcal{Z}})}).
\]
By the properties of the $\operatorname{unvec}(.)$ operator and the definition of a matrix zonotope, we have $\mathcal{M} = \operatorname{unvec}(\zono{c_{\mathcal{Z}}, G_{\mathcal{Z}}})$ and 
\[\zono{C_{\mathcal{M}}, \operatorname{unvec}(\downarrow_q(G_{\mathcal{Z}}))} = \operatorname{unvec}(\zono{c_{\mathcal{Z}}, \downarrow_q(G_{\mathcal{Z}})}).\]
Defining the reduced generator matrix for $\mathcal{M}$ as $\downarrow_q(\tilde{G}_{\mathcal{M}}) = \operatorname{unvec}(\downarrow_q(G_{\mathcal{Z}}))$, the above containment simplifies to the desired result
\[
\mathcal{M} \subseteq \zono{C_{\mathcal{M}},\downarrow_q(\tilde{G}_{\mathcal{M}})}.
\]
\end{proof}
\begin{property}[\textbf{Matrix Derivative} \cite{petersen2008matrix}]\label{pr:matrixdrevative}
For given matrices $A, B, C, X$ with conformable dimensions, the following derivative rules hold:
\begin{enumerate}
    \item $ \frac{\partial \text{Tr}(AX^\top)}{\partial X} = A,$
    \item $ \frac{\partial \text{Tr}(AXBX^\top C)}{\partial X} = A^\top C^\top XB^\top + CAXB.$
\end{enumerate}
\end{property}

%% file: Sections/Appendix_B.tex
\section*{Appendix B: Lemmas and Proofs}

\begin{lemma}[\textbf{Matrix Perturbation}]\label{le:matrix_upper_correctedz_inf}
Let $C_{\mathcal{M}_p,0} \in \mathbb{R}^{n_{\mathcal{M}_p} \times m_{\mathcal{M}_p}}$ be an initial nominal matrix. Consider a sequence of matrices $\{C_{\mathcal{M}_p,j}\}_{j=0}^k$ where $C_{\mathcal{M}_p,j+1} = C_{\mathcal{M}_p,j} + \delta C_{\mathcal{M}_p,j}$. By defining a perturbation matrix $\delta C_{\mathcal{M}_p,j} \in \mathbb{R}^{n_{\mathcal{M}_p} \times m_{\mathcal{M}_p}}$ which is bounded by a constant $\mu > 0$ with respect to the element-wise infinity norm, i.e., $\|\delta C_{\mathcal{M}_p,j}\|_{max} \leq \mu$.

Then, the matrix $C_{\mathcal{M}_p,k}$ (after $k$ steps) is contained within the matrix zonotope
\[
C_{\mathcal{M}_p,k} \in \mathcal{M}_{p_k}.
\]
where
\[
\mathcal{M}_{{p_k}} = \zono{C_{\mathcal{M}_p,0}, \check{G}_{\mathcal{M}_{p_k}}}.
\]
with the generator set $\check{G}_{\mathcal{M}_{p_k}} = \Big\{ G^{(1)}_{\mathcal{M}_{p_k}} \dots G^{(\gamma_{\mathcal{M}_p})}_{\mathcal{M}_{p_k}}\Big\}$, where $\gamma_{\mathcal{M}_p} := n_{\mathcal{M}_p} \cdot m_{\mathcal{M}_p}$ and the generators are defined as
\[
G^{(s)}_{\mathcal{M}_{p_k}} := k\mu E_s.
\]
and $E_s$ are the standard basis matrices of size $n_{\mathcal{M}_p} \times m_{\mathcal{M}_p}$ with a single non-zero element of 1. Here, the index $s$ corresponds to the flattened index of the matrix elements.
\end{lemma}

\begin{proof}
The matrix $C_{\mathcal{M}_p,k}$ is given by the sum of successive perturbations
\[
C_{\mathcal{M}_p,k} = C_{\mathcal{M}_p,0} + \sum_{j=0}^{k-1} \delta C_{\mathcal{M}_p,j}.
\]
The perturbation matrix $\delta C_{\mathcal{M}_p,j}$ is bounded by $\|\delta C_{\mathcal{M}_p,j}\|_\infty \leq \mu$, which directly implies that for every element $(i,l)$ of the matrix, $|(\delta C_{\mathcal{M}_p,j})_{(i,l)}| \leq \mu$.

Now, for any element $(i,l)$
\[
\left(\sum_{j=0}^{k-1} \delta C_{\mathcal{M}_p,j}\right)_{(i,l)} = \sum_{j=0}^{k-1} (\delta C_{\mathcal{M}_p,j})_{(i,l)}.
\]
Since $|(\delta C_{\mathcal{M}_p,j})_{(i,l)}| \leq \mu$ for each term in the sum, the sum of these elements is bounded as:
\[
\left|\sum_{j=0}^{k-1} (\delta C_{\mathcal{M}_p,j})_{(i,l)}\right| \leq \sum_{j=0}^{k-1} |(\delta C_{\mathcal{M}_p,j})_{(i,l)}| \leq \sum_{j=0}^{k-1} \mu = k\mu.
\]
Thus, the sum of perturbations is a matrix whose elements are bounded by $k\mu$. Let's call the sum of perturbations $\Delta C_{\mathcal{M}_p}$. Each element $(\Delta C_{\mathcal{M}_p})_{(i,l)}$ satisfies $|(\Delta C_{\mathcal{M}_p})_{(i,l)}| \leq k\mu$. We can express $\Delta C_{\mathcal{M}_p}$ as a linear combination of standard basis matrices:
\[
\Delta C_{\mathcal{M}_p} = \sum_{i=1}^{n_{\mathcal{M}_p}}\sum_{l=1}^{m_{\mathcal{M}_p}} (\Delta C_{\mathcal{M}_p})_{(i,l)} E_{i,l}.
\]
By using the flattened index $s$ for the pair $(i,l)$, we can rewrite this sum as
\[
\Delta C_{\mathcal{M}_p} = \sum_{s=1}^{\gamma_{\mathcal{M}_p}} c_s E_s,
\]
where $E_s$ is a standard basis matrix and $|c_s| \le k\mu$. We can define a set of coefficients $|\beta| \le 1$ such that $c_s = k\mu {\beta}$. Substituting this back into the sum, we get:
\[
\Delta C_{\mathcal{M}_p} = \sum_{s=1}^{\gamma_{\mathcal{M}_p}}{\beta} (k\mu E_s).
\]
This is precisely the definition of the matrix zonotope with center zero and generators $k\mu E_s$. Since $C_{\mathcal{M}_p,k} = C_{\mathcal{M}_p,0} + \Delta C_{\mathcal{M}_p}$, the matrix $C_{\mathcal{M}_p,k}$ is contained within the matrix zonotope $\mathcal{M}_{p_k}$ with center $C_{\mathcal{M}_p,0}$ and the defined generators.
\end{proof}